\documentclass[12pt,oneside]{amsart}

\usepackage{amsaddr}
\usepackage{amsthm, amsmath, amssymb, amsbsy, amscd, amsfonts,url,
}

\newtheorem{thm}{Theorem}[section]

\allowdisplaybreaks
\newtheorem{proposition}[thm]{Proposition}
\newtheorem{corollary}[thm]{Corollary}

\newtheorem{theorem}[thm]{Theorem}

\theoremstyle{definition}
\newtheorem{definition}[thm]{Definition}
\newtheorem{example}[thm]{Example}
\newtheorem{remark}[thm]{Remark}

\newcommand{\cprime}{\/{\mathsurround=0pt$'$}}
\newcommand{\md}{W}
\newcommand{\Mfd}{\mathbf{M}}
\newcommand{\cl}{\colon}
\newcommand{\dft}{\mathrm{d}}
\newcommand{\qqquad}{\qquad\quad}

\newcommand{\ym}{\mathbb{Y}}
\newcommand{\tm}{\mathbb{T}}

\newcommand{\zsp}{\mathbb{Z}_{>0}}

\DeclareMathOperator{\mat}{Mat}

\DeclareMathOperator{\id}{Id}

\newcommand{\xc}{\mathbf{x}}
\newcommand{\ct}{c}
\newcommand{\qt}{q}
\newcommand{\bm}{m}
\newcommand{\dm}{N}
\newcommand{\dmp}{k}

\newcommand{\gmd}{\mathbf{G}}
\newcommand{\mA}{\mathbf{A}}
\newcommand{\mB}{\mathbf{B}}
\newcommand{\mC}{\mathbf{C}}
\newcommand{\mD}{\mathbf{D}}
\newcommand{\mL}{\mathbf{L}}

\newcommand{\bC}{\mathbb{C}}
\newcommand{\bR}{\mathbb{R}}
\newcommand{\bK}{\mathbb{K}}

\newcommand{\ah}{\mathbf{a}}

\newcommand{\ddh}{\mathbf{d}}
\newcommand{\bh}{\mathbf{b}}
\newcommand{\ch}{\mathbf{c}}

\newcommand{\fp}{\mathrm{p}}
\newcommand{\sz}{n}

\title[Set-theoretical solutions to the Zamolodchikov tetrahedron equation]%
{Algebraic and differential-geometric constructions of set-theoretical solutions to the Zamolodchikov tetrahedron equation}
\date{}

\author{Sergei Igonin\qquad\qquad Sotiris Konstantinou-Rizos} 
\address{Centre of Integrable Systems, P.G. Demidov Yaroslavl State University, Yaroslavl, Russia}

\author{}
\email{s-igonin@yandex.ru, skonstantin84@gmail.com}

\begin{document}

\maketitle

\begin{abstract}
We present several algebraic and differential-geometric constructions of 
tetrahedron maps, which are set-theoretical solutions to the Zamolodchikov tetrahedron equation.
In particular, we obtain a family of new (nonlinear) polynomial tetrahedron maps 
on the space of square matrices of arbitrary size, using a matrix refactorisation equation, 
which does not coincide with the standard local Yang--Baxter equation.
Liouville integrability is established for some of these maps.

Also, we show how to derive linear tetrahedron maps 
as linear approximations of nonlinear ones, using Lax representations and the differentials 
of nonlinear tetrahedron maps on manifolds.
We apply this construction to two nonlinear maps: 
a tetrahedron map obtained in~\cite{Dimakis}
in a study of soliton solutions of vector KP equations and 
a tetrahedron map obtained in~\cite{KR}
in a study of a matrix trifactorisation problem 
related to a Darboux matrix associated with a Lax operator for the NLS equation.
We derive parametric families of new linear tetrahedron maps 
(with nonlinear dependence on parameters), 
which are linear approximations for these nonlinear ones.

Furthermore, we present (nonlinear) matrix generalisations
of a tetrahedron map from Sergeev's classification~\cite{Sergeev}. 
These matrix generalisations can be regarded as tetrahedron maps in noncommutative variables.

Besides, several tetrahedron maps on arbitrary groups are constructed.
\end{abstract}

\bigskip

\hspace{.2cm} \textbf{PACS numbers:} 02.30.Ik, 02.90.+p, 03.65.Fd.

\hspace{.2cm} \textbf{Mathematics Subject Classification 2020:} 16T25, 81R12.


\hspace{.2cm} \textbf{Keywords:} 
Zamolodchikov tetrahedron equation, quantum Yang--Baxter equation, 

\hspace{2.6cm} tetrahedron maps, Yang--Baxter maps, Liouville integrability, Lax representations, 

\hspace{2.6cm} differentials of maps.

\section{Introduction}

\label{sintr}

The Zamolodchikov tetrahedron equation~\cite{Zamolodchikov,Zamolodchikov-2}
is a higher-dimensional analogue of the well-celebrated quantum Yang--Baxter equation.
They belong to the most fundamental equations in mathematical physics 
and have applications in many diverse branches of physics and mathematics, 
including statistical mechanics, quantum field theories, algebraic topology, 
and the theory of integrable systems. 
The Yang--Baxter and tetrahedron equations are members of the family 
of $n$-simplex equations~\cite{bazh82,DMH15,Hietarinta97,Sharygin-Talalaev,Nijhoff-2,Nijhoff},
where they correspond to the cases of $2$-simplex and $3$-simplex, respectively.
Some applications of the tetrahedron equation can be found 
in~\cite{Bazhanov-Sergeev,Bazhanov-Sergeev-2,DMH15,Doliwa-Kashaev,Gorbounov-Talalaev,Kapranov,Kashaev-Sergeev,Kassotakis-Tetrahedron,Nijhoff-2,Nijhoff,TalalUMN21,Yoneyama21} and references therein.

This paper is devoted to tetrahedron maps and their relations with Yang--Baxter maps, 
which are set-theoretical solutions to the tetrahedron equation 
and the Yang--Baxter equation, respectively.
Set-theoretical solutions of the Yang--Baxter equation have been intensively studied
by many authors after the work of Drinfeld~\cite{Drin92}.
Even before that, examples of such solutions were constructed by Sklyanin~\cite{skl88}.
A quite general construction for tetrahedron maps first appeared 
in works of Korepanov (see~\cite{Korep94,Korepanov-DS} and references therein) 
in connection with integrable dynamical systems in discrete time.
Presently, the relations of tetrahedron maps and Yang--Baxter maps
with integrable systems (including PDEs and lattice equations) 
and with algebraic structures (including groups and rings) are very active areas of research 
(see, e.g.,~\cite{abs2003,atk14,Bard_talk_10-2021,Vincent,carter2006,Dimakis,Doliwa-Kashaev,Ito2013,Kashaev-Sergeev,Kassotakis-Tetrahedron,KR,Sokor-Sasha,Kouloukas2,pap-Tongas,TalalUMN21,Yoneyama21}).
In this paper we present several algebraic and differential-geometric constructions of 
tetrahedron maps. 
Some of our constructions of tetrahedron maps use Yang--Baxter maps as an auxiliary tool.
The paper is organised as follows.

Section~\ref{preliminaries} contains the definitions  
of Yang--Baxter and tetrahedron maps and recalls some basic properties of them.
Section~\ref{stmlyb} recalls the well-known method
to construct tetrahedron maps by means of the matrix local Yang--Baxter equation 
written as a matrix refactorisation~\cite{Kashaev-Sergeev,Sergeev,Doliwa-Kashaev,KR}.
In Section~\ref{stmomf} we present a modification of this method.
Namely, in Section~\ref{stmomf} we find new tetrahedron maps by means of 
other matrix refactorisations, which 
are similar to the local Yang--Baxter equation, but do not coincide with it.
The following notation is used.
\begin{itemize}
	\item $\bK$ is a field. 
	Usually, $\bK$ is either $\bR$ or $\bC$, but one can consider also arbitrary fields.
	\item $\zsp$ is the set of positive integers.
	\item $\mathbf{1}_n$ is the $n\times n$ identity matrix for $n\in\zsp$.
	\item $\mat_n(\bK)$ is the set of $n\times n$ matrices with entries from~$\bK$.
\end{itemize}
For any $\sz,\bm\in\zsp$ we consider the matrix refactorisation equation~\eqref{mrfe} with~\eqref{L12bm},
where $\mA_{\sz}$, $\mB_{\sz}$, $\mC_{\sz}$, $\mD_{\sz}$ are $\sz\times\sz$ matrix-functions on a set~$\md$.
(That is, $\mA_{\sz}$, $\mB_{\sz}$, $\mC_{\sz}$, $\mD_{\sz}$ are maps from $\md$ to $\mat_{\sz}(\bK)$.)
In the case $\bm\neq\sz$ the matrices~\eqref{L12bm} are different from~\eqref{L1213}, 
and, therefore, equation~\eqref{mrfe} does not coincide 
with the standard matrix local Yang--Baxter equation~\eqref{lybef}.

Using equation~\eqref{mrfe} with~\eqref{L12bm} in the case when $\md=\mat_{\sz}(\bK)$ and
\begin{gather}
\label{iabcd}
\begin{gathered}
\mA_{\sz},\mB_{\sz},\mC_{\sz},\mD_{\sz}\cl\mat_{\sz}(\bK)\to\mat_{\sz}(\bK),\\
\mA_{\sz}(X)=\mathbf{1}_{\sz},\qquad\mB_{\sz}(X)=X,\qquad\mC_{\sz}(X)=0,\qquad\mD_{\sz}(X)=\mathbf{1}_{\sz},\qquad
X\in\mat_{\sz}(\bK),
\end{gathered}
\end{gather}
we obtain the following new polynomial tetrahedron maps:
\begin{itemize}
	\item \eqref{mc1} for $\bm\in\{1,\dots,\sz\}$,
	\item \eqref{mc2} for $\bm\in\{\sz+1,\dots,2\sz-1\}$.
\end{itemize}
In our opinion, it is interesting to see that such a nonstandard matrix refactorisation equation
gives tetrahedron maps.
The recent paper~\cite{igon322} (which was written later than the first version of this paper) 
presents an algebraic explanation of the fact that the maps~\eqref{mc1},~\eqref{mc2}
satisfy the tetrahedron equation.

In Section~\ref{slint} we prove Liouville integrability 
\begin{itemize}
	\item for the map~\eqref{mc1} in the case $1\le\bm\le\sz/2$ in Theorem~\ref{emnmc1},
	\item for the map~\eqref{mc2} in the case $3\sz/2\le\bm\le 2\sz-1$ in Theorem~\ref{emnmc2}.
\end{itemize}
Examples~\ref{en2m1}, \ref{en5m2}, \ref{en2m3}, \ref{en5m8} 
clarify these results for small values of $\sz$, $\bm$.

We use the standard notion of Liouville integrability for maps on manifolds
(see, e.g.,~\cite{fordy14,Sokor-Sasha,vesel1991} and references therein).
This notion is presented in Definition~\ref{dli}.
Recall that Liouville integrability plays essential roles 
in classical mechanics and in the theory of integrable maps
(see, e.g.,~\cite{fordy14,hjn-book,Pavlos2019,Sokor-Sasha,vesel1991}).
We think that new examples of Liouville integrable maps satisfying the tetrahedron equation
make valuable contribution in clarifying the interplay between the classical Liouville integrability 
and the $3$-dimensional discrete integrability encoded in the tetrahedron equation.

In Section~\ref{stmlyb} we construct also new tetrahedron 
maps~\eqref{mnK},~\eqref{tmnK}, 
which are matrix generalisations of the map~\eqref{msk} 
from Sergeev's classification~\cite{Sergeev}.
As shown in Proposition~\ref{rnc},
in~\eqref{mnK},~\eqref{tmnK} one can 
replace~$\mat_{\sz}(\bK)$ by an arbitrary associative ring~$\mathcal{A}$ with a unit.
As a result, we get the maps~\eqref{amnK},~\eqref{atmnK}, 
which can be regarded as tetrahedron maps in noncommutative variables.

It is known that the local Yang--Baxter equation~\cite{Nijhoff-2} can be viewed  
as a ``Lax equation'' or ``Lax system'' for the tetrahedron equation 
(see, e.g.,~\cite{DMH15} and references therein).
This allows one to introduce the notion of Lax representations 
for tetrahedron maps, see Remark~\ref{rlaxr} for details.

In Section~\ref{sdtm} we show how to derive linear tetrahedron maps 
as linear approximations of nonlinear ones, using Lax representations and 
results of~\cite{ikkp21} on the differentials of nonlinear tetrahedron maps on manifolds.
In our opinion, this connection between Lax representations 
and the construction of linear tetrahedron maps 
as linear approximations of nonlinear ones is interesting and deserves attention.
(As discussed in Remark~\ref{rresul}, 
Lax representations were not considered in~\cite{ikkp21}.)
In Examples~\ref{edmh},~\ref{enls} we apply this construction to two nonlinear maps: 
\begin{itemize}
	\item a tetrahedron map obtained in~\cite{Dimakis}
in a study of soliton solutions of vector Kadomtsev--Petviashvili (KP) equations;
\item a tetrahedron map obtained in~\cite{KR} in a study of a matrix trifactorisation problem 
related to a Darboux matrix associated with a Lax operator for the nonlinear Schr\"odinger (NLS) equation.
\end{itemize}
As a result, 
one derives parametric families of new linear tetrahedron maps~\eqref{lmqt},~\eqref{lmqt1},~\eqref{lmqt2}, 
which are linear approximations for these nonlinear ones.
The obtained parametric families of linear tetrahedron maps~\eqref{lmqt},~\eqref{lmqt1},~\eqref{lmqt2} 
depend nonlinearly on the parameters $\qt$, $\qt_1$, $\qt_2$.

Relations with results of~\cite{ikkp21} are discussed in Remark~\ref{rresul}.
Note that linear approximations of nonlinear Yang--Baxter maps
were considered in~\cite{Sokor-Sasha,BIKRP}.

In Section~\ref{stmgr} we construct new tetrahedron maps~\eqref{tmgr} 
on an arbitrary group~$\gmd$. 
Section~\ref{sconc} concludes the paper 
with suggestions on how the results of this paper can be extended.

\section{Preliminaries}
\label{preliminaries}

For any set $S$ and $n\in\zsp$, we use the notation
$S^n=\underbrace{S\times S\times\dots\times S}_{n}$.
Also, we denote by $\id_S\cl S\to S$ the identity map.

Let $\md$ be a set. A \emph{Yang--Baxter map} is a map 
$$
Y\colon \md\times \md\to \md\times \md,\qquad Y(x,y)=\big(u(x,y),v(x,y)\big),\qquad x,y\in \md,
$$
satisfying the Yang--Baxter equation
\begin{equation}\label{eq_YB}
Y^{12}\circ Y^{13}\circ Y^{23}=Y^{23}\circ Y^{13}\circ Y^{12}.
\end{equation}
The terms $Y^{12}$, $Y^{13}$, $Y^{23}$ in~\eqref{eq_YB} are maps $\md^3\to \md^3$ defined as follows 
\begin{gather*}
Y^{12}(x,y,z)=\big(u(x,y),v(x,y),z\big),\qquad\quad
Y^{23}(x,y,z)=\big(x,u(y,z),v(y,z)\big),\\
Y^{13}(x,y,z)=\big(u(x,z),y,v(x,z)\big),\qquad\quad x,y,z\in \md.
\end{gather*}

A \emph{tetrahedron map} is a map 
\begin{equation}
\notag
T\cl \md^3\rightarrow \md^3,\qquad
T(x,y,z)=\big(f(x,y,z),g(x,y,z),h(x,y,z)\big),\qquad x,y,z\in \md,
\end{equation}
satisfying the (Zamolodchikov) \emph{tetrahedron equation}
\begin{equation}\label{tetr-eq}
    T^{123}\circ T^{145} \circ T^{246}\circ T^{356}=T^{356}\circ T^{246}\circ T^{145}\circ T^{123}.
\end{equation}
Here $T^{ijk}\cl \md^6\rightarrow \md^6$ for $i,j,k=1,\ldots,6$, $i<j<k$, is the map 
acting as $T$ on the $i$th, $j$th, $k$th factors 
of the Cartesian product $\md^6$ and acting as identity on the remaining factors.
For instance,
$$
T^{246}(x,y,z,r,s,t)=\big(x,f(y,r,t),z,g(y,r,t),s,h(y,r,t)\big),\qqquad x,y,z,r,s,t\in \md.
$$

Sometimes we write $T\big((x,y,z)\big)$ instead of $T(x,y,z)$, 
in order to emphasise that we apply $T$ to a given point $(x,y,z)\in\md^3$.

\begin{proposition}[\cite{Kassotakis-Tetrahedron}]
\label{pp13}
Consider the permutation map
\begin{gather*}
P^{13}\cl \md^3\to \md^3,\qqquad
P^{13}(a_1,a_2,a_3)=(a_3,a_2,a_1),\qquad a_i\in \md.
\end{gather*}
If a map $T\cl \md^3\to \md^3$ satisfies the tetrahedron equation~\eqref{tetr-eq}
then $\tilde{T}=P^{13}\circ T\circ P^{13}$ obeys this equation as well.
\end{proposition}

\begin{proposition}[\cite{Kassotakis-Tetrahedron}]
\label{pist}
Let $T\cl \md^3\to \md^3$ be a tetrahedron map.
Suppose that a map $\sigma\cl \md\to \md$ satisfies 
\begin{gather}
\label{stss}
(\sigma\times\sigma\times\sigma)\circ T\circ(\sigma\times\sigma\times\sigma)=T,
\qqquad \sigma\circ\sigma=\id_\md.
\end{gather}
Note that, since $\sigma\circ\sigma=\id_\md$, 
the relation $(\sigma\times\sigma\times\sigma)\circ T\circ(\sigma\times\sigma\times\sigma)=T$
is equivalent to $(\sigma\times\sigma\times\sigma)\circ T=T\circ (\sigma\times\sigma\times\sigma)$, 
hence~\eqref{stss} means that $\sigma$ is an involutive symmetry of $T$.

Then
\begin{gather*}
\tilde{T}=(\sigma\times\id_\md\times\sigma)\circ T\circ(\id_\md\times\sigma\times\id_\md),\qqquad
\hat{T}=(\id_\md\times\sigma\times\id_\md)\circ T\circ(\sigma\times\id_\md\times\sigma)
\end{gather*}
are tetrahedron maps.
\end{proposition}

The statement of Proposition~\ref{prop_YB_tetr} is well known.
A proof can be found, e.g., in~\cite{ikkp21}.
\begin{proposition}\label{prop_YB_tetr}
Let $Y\colon \md^2\rightarrow \md^2$ be a Yang--Baxter map. 
Then the maps
$$
Y^{12}=Y\times\id_\md\colon \md^3\rightarrow \md^3,\qqquad
Y^{23}=\id_\md\times Y\colon \md^3\rightarrow \md^3
$$ 
are tetrahedron maps.
\end{proposition}

The result of Proposition~\ref{prinvt} is also well known,
but for completeness we present a proof for it.
\begin{proposition}
\label{prinvt}
Let $T\cl \md^3\to \md^3$ be an invertible map satisfying the
tetrahedron equation~\eqref{tetr-eq}.
Then the inverse map $T^{-1}\cl \md^3\to \md^3$ obeys this equation as well.
\end{proposition}
\begin{proof}
We set $\tilde{T}=T^{-1}$. It is easily seen that
\begin{gather}
\label{ttt}
\tilde{T}^{123}=(T^{123})^{-1},\qquad
\tilde{T}^{145}=(T^{145})^{-1},\qquad
\tilde{T}^{246}=(T^{246})^{-1},\qquad
\tilde{T}^{356}=(T^{356})^{-1}.
\end{gather}
Using~\eqref{ttt}, one obtains
\begin{gather}
\label{tt1}
(T^{123}\circ T^{145}\circ T^{246}\circ T^{356})^{-1}=
(T^{356})^{-1}(T^{246})^{-1}(T^{145})^{-1}(T^{123})^{-1}=
\tilde{T}^{356}\circ \tilde{T}^{246}\circ \tilde{T}^{145}\circ \tilde{T}^{123},\\
\label{tt2}
(T^{356}\circ T^{246}\circ T^{145}\circ T^{123})^{-1}=
(T^{123})^{-1}(T^{145})^{-1}(T^{246})^{-1}(T^{356})^{-1}=
\tilde{T}^{123}\circ\tilde{T}^{145}\circ\tilde{T}^{246}\circ\tilde{T}^{356}.
\end{gather}
From~\eqref{tetr-eq},~\eqref{tt1},~\eqref{tt2} we derive 
$\tilde{T}^{123}\circ\tilde{T}^{145}\circ\tilde{T}^{246}\circ\tilde{T}^{356}=
\tilde{T}^{356}\circ \tilde{T}^{246}\circ \tilde{T}^{145}\circ \tilde{T}^{123}$,
which means that the map $\tilde{T}=T^{-1}$ satisfies the tetrahedron equation.
\end{proof}

\section{Tetrahedron maps related to matrix refactorisations}
\label{stmrmr}

\subsection{Tetrahedron maps associated with the matrix local Yang--Baxter equation} 
\label{stmlyb}

In this subsection we construct new tetrahedron maps by means 
of the well-known method based on the matrix local Yang--Baxter equation 
written as a matrix refactorisation~\cite{Kashaev-Sergeev,Sergeev,Doliwa-Kashaev,KR}.
In Subsection~\ref{stmomf} we find new tetrahedron maps by means of 
other matrix refactorisations, which 
are similar to the local Yang--Baxter equation, but do not coincide with it.

Let $\bK$ be a field.
For any $n\in\zsp$ we denote by $\mat_n(\bK)$ 
the set of $n\times n$ matrices with entries from~$\bK$.
Let $\mathbf{1}_n$ be the $n\times n$ identity matrix.

Consider a set $\md$ and maps 
$\mA_{\sz},\mB_{\sz},\mC_{\sz},\mD_{\sz}\cl \md\to\mat_{\sz}(\bK)$ for some ${\sz}\in\zsp$.
That is, $\mA_{\sz},\mB_{\sz},\mC_{\sz},\mD_{\sz}$ are ${\sz}\times {\sz}$ matrix-functions on the set $\md$.
Then we have the corresponding $2{\sz}\times 2{\sz}$ matrix-function
$$
\mL=\begin{pmatrix}
     \mA_{\sz}& \mB_{\sz}  \\
     \mC_{\sz} & \mD_{\sz}
\end{pmatrix}\cl \md\to\mat_{2{\sz}}(\bK)
$$ 
and the $3{\sz}\times 3{\sz}$ matrix-functions 
$\mL^{12},\mL^{13},\mL^{23}\cl \md\to\mat_{3{\sz}}(\bK)$
\begin{gather}
\label{L1213}
\mL^{12}=\begin{pmatrix}
     \mA_{\sz}& \mB_{\sz} &  0 \\
     \mC_{\sz} & \mD_{\sz} & 0 \\
		0 & 0 & \mathbf{1}_{\sz}
\end{pmatrix},\qqquad
\mL^{13}=\begin{pmatrix}
     \mA_{\sz}& 0 &  \mB_{\sz} \\
     0 & \mathbf{1}_{\sz} & 0 \\
		\mC_{\sz} & 0 & \mD_{\sz}
\end{pmatrix},\qqquad
\mL^{23}=\begin{pmatrix}
    \mathbf{1}_{\sz} &0 &0\\ 
		0& \mA_{\sz}& \mB_{\sz}  \\
    0& \mC_{\sz} & \mD_{\sz} 
\end{pmatrix}.
\end{gather}
Consider the corresponding local Yang--Baxter equation
\begin{equation} 
\label{lybef}
\mL^{12} (x) \cdot\mL^{13} (y)\cdot\mL^{23}(z) = 
\mL^{23}(\hat{z})\cdot\mL^{13}(\hat{y})\cdot\mL^{12}(\hat{x}),\qqquad
x,y,z,\hat{x},\hat{y},\hat{z}\in \md.
\end{equation}

To our knowledge, in all available examples the following property is valid:

\noindent
\emph{Suppose that equation}~\eqref{lybef} \emph{determines a map}
\begin{gather}
\label{mflyb}
\tm\cl \md\times\md\times\md\to\md\times\md\times\md,\qqquad
\tm\big((x,y,z)\big)=(\hat{x},\hat{y},\hat{z}),
\end{gather}
\emph{in the sense that equation}~\eqref{lybef} \emph{is equivalent to the relation}
$(\hat{x},\hat{y},\hat{z})=\tm\big((x,y,z)\big)$.
\emph{Then the map~\eqref{mflyb} satisfies the tetrahedron equation}~\eqref{tetr-eq}.

One can expect that this property holds in general under some non-degeneracy conditions,
but we do not find a detailed general proof in the available literature.
When we get a map~\eqref{mflyb} arising from~\eqref{lybef}, 
we check separately that the map satisfies the tetrahedron equation.

\begin{example}
\label{esmk}
Let $\sz=1$ and $\md=\bK$. Following~\cite{Sergeev}, 
fix a constant $k\in\bK$ and consider the $2\times 2$ matrix-function
$$
\mL\cl \bK\to\mat_{2}(\bK),\qqquad \mL(x)=
\begin{pmatrix}
     1& x  \\
     0 & k
\end{pmatrix},\qqquad x\in\bK.
$$ 
As shown in~\cite{Sergeev}, 
the corresponding local Yang--Baxter equation~\eqref{lybef} 
determines the tetrahedron map
\begin{gather}
\label{msk}
\tm\cl \bK^3\to\bK^3,\qqquad
(x,y,z)\mapsto(\hat{x},\hat{y},\hat{z})=(x,\,ky+xz,\,z).
\end{gather}
As noticed in~\cite{Sergeev}, if $k\neq 0$, then the map~\eqref{msk} 
is invertible, and the inverse map 
\begin{gather}
\label{invms}
(\tm)^{-1}\cl \bK^3\to\bK^3,\qqquad
(\hat{x},\hat{y},\hat{z})\mapsto(x,y,z)=
\Big(\hat{x},\,\frac{1}{k}(\hat{y}-\hat{x}\hat{z}),\,\hat{z}\Big),
\end{gather}
satisfies the tetrahedron equation as well.

Now let $\sz\in\zsp$ and $\md=\mat_{\sz}(\bK)$.
To construct matrix generalisations of the map~\eqref{msk}, 
fix a nondegenerate $\sz\times\sz$ matrix $K\in\mat_{\sz}(\bK)$
and consider the $2\sz\times 2\sz$ matrix-function
$$
\mL_K\cl\mat_{\sz}(\bK)\to\mat_{2\sz}(\bK),\qqquad \mL_K(X)=
\begin{pmatrix}
     \mathbf{1}_{\sz}& X  \\
     0 & K
\end{pmatrix},\qqquad X\in\mat_{\sz}(\bK).
$$ 
We have 
\begin{gather*}
\mL_K^{12} (X) \cdot\mL_K^{13} (Y)\cdot\mL_K^{23}(Z)=
\begin{pmatrix}
     \mathbf{1}_{\sz} & X &  0 \\
     0 & K & 0 \\
		0 & 0 & \mathbf{1}_{\sz}
\end{pmatrix}
\begin{pmatrix}
     \mathbf{1}_{\sz} & 0 &  Y \\
     0 & \mathbf{1}_{\sz} & 0 \\
		0 & 0 & K
\end{pmatrix}
\begin{pmatrix}
    \mathbf{1}_{\sz} &0 &0\\ 
		0& \mathbf{1}_{\sz} & Z  \\
    0& 0 & K 
\end{pmatrix}=\begin{pmatrix}
    \mathbf{1}_{\sz} &X &XZ+YK\\ 
		0& K & KZ  \\
    0& 0 & K^2 
\end{pmatrix},\\
\mL_K^{23}(\hat{Z})\cdot\mL_K^{13}(\hat{Y})\cdot\mL_K^{12}(\hat{X})=
\begin{pmatrix}
    \mathbf{1}_{\sz} &0 &0\\ 
		0& \mathbf{1}_{\sz} & \hat{Z}  \\
    0& 0 & K 
\end{pmatrix}
\begin{pmatrix}
     \mathbf{1}_{\sz} & 0 &  \hat{Y} \\
     0 & \mathbf{1}_{\sz} & 0 \\
		0 & 0 & K
\end{pmatrix}
\begin{pmatrix}
     \mathbf{1}_{\sz} & \hat{X} &  0 \\
     0 & K & 0 \\
		0 & 0 & \mathbf{1}_{\sz}
\end{pmatrix}
=\begin{pmatrix}
     \mathbf{1}_{\sz} & \hat{X} &  \hat{Y} \\
     0 & K & \hat{Z}K \\
		0 & 0 & K^2
\end{pmatrix},\\
X,Y,Z,\hat{X},\hat{Y},\hat{Z}\in\mat_{\sz}(\bK).
\end{gather*}
The local Yang--Baxter equation 
$\mL_K^{12} (X) \cdot\mL_K^{13} (Y)\cdot\mL_K^{23}(Z)=
\mL_K^{23}(\hat{Z})\cdot\mL_K^{13}(\hat{Y})\cdot\mL_K^{12}(\hat{X})$
gives the map 
\begin{gather}
\label{mnK}
\tm_K\cl\big(\mat_{\sz}(\bK)\big)^3\to\big(\mat_{\sz}(\bK)\big)^3,\quad
(X,Y,Z)\mapsto(\hat{X},\hat{Y},\hat{Z})=(X,\,YK+XZ,\,KZK^{-1}). 
\end{gather}
It is easy to verify that the map~\eqref{mnK} satisfies the tetrahedron equation.
By Proposition~\ref{prinvt}, the inverse map 
\begin{gather}
\label{invmK}
\begin{gathered}
(\tm_K)^{-1}\cl\big(\mat_{\sz}(\bK)\big)^3\to\big(\mat_{\sz}(\bK)\big)^3,\\
(\hat{X},\hat{Y},\hat{Z})\mapsto (X,Y,Z)=
(\hat{X},\,\hat{Y}K^{-1}-\hat{X}K^{-1}\hat{Z},\,K^{-1}\hat{Z}K),
\end{gathered}
\end{gather}
obeys the tetrahedron equation as well.

Now consider the permutation map $P^{13}\cl\big(\mat_{\sz}(\bK)\big)^3\to\big(\mat_{\sz}(\bK)\big)^3$,
$\,(X,Y,Z)\mapsto(Z,Y,X)$.
Applying Proposition~\ref{pp13} to~\eqref{mnK},
we see that the map 
\begin{gather}
\label{tmnK}
\begin{gathered}
\tilde{\tm}_K=P^{13}\circ \tm_K\circ P^{13}
\cl\big(\mat_{\sz}(\bK)\big)^3\to\big(\mat_{\sz}(\bK)\big)^3,\\
(X,Y,Z)\mapsto(\hat{X},\hat{Y},\hat{Z})=(KXK^{-1},\,YK+ZX,\,Z), 
\end{gathered}
\end{gather}
also satisfies the tetrahedron equation. By Proposition~\ref{prinvt}, the inverse map 
\begin{gather}
\label{tinvmK}
\begin{gathered}
(\tilde{\tm}_K)^{-1}\cl\big(\mat_{\sz}(\bK)\big)^3\to\big(\mat_{\sz}(\bK)\big)^3,\\
(\hat{X},\hat{Y},\hat{Z})\mapsto (X,Y,Z)=
(K^{-1}\hat{X}K,\,\hat{Y}K^{-1}-\hat{Z}K^{-1}\hat{X},\,\hat{Z}),
\end{gathered}
\end{gather}
obeys the tetrahedron equation as well.
The maps~\eqref{mnK},~\eqref{tmnK} are matrix generalisations of~\eqref{msk}.
\end{example}

Let $\mathcal{A}$ be an associative ring with a unit.
Let $K\in\mathcal{A}$ be an invertible element.
Replacing $\mat_{\sz}(\bK)$ by~$\mathcal{A}$ 
in~\eqref{mnK}--\eqref{tinvmK}, we obtain the following result.
\begin{proposition}
\label{rnc}
For any associative ring $\mathcal{A}$ with a unit and any invertible element 
$K\in\mathcal{A}$, we have the tetrahedron maps
\begin{gather}
\label{amnK}
\tm_K\cl(\mathcal{A})^3\to(\mathcal{A})^3,\qqquad
(X,Y,Z)\mapsto(\hat{X},\hat{Y},\hat{Z})=(X,\,YK+XZ,\,KZK^{-1}),\\
\label{ainvmK}
(\tm_K)^{-1}\cl(\mathcal{A})^3\to(\mathcal{A})^3,\quad
(\hat{X},\hat{Y},\hat{Z})\mapsto (X,Y,Z)=
(\hat{X},\,\hat{Y}K^{-1}-\hat{X}K^{-1}\hat{Z},\,K^{-1}\hat{Z}K),\\
\label{atmnK}
\tilde{\tm}_K\cl(\mathcal{A})^3\to(\mathcal{A})^3,\qqquad
(X,Y,Z)\mapsto(\hat{X},\hat{Y},\hat{Z})=(KXK^{-1},\,YK+ZX,\,Z),\\
\label{atinvmK}
(\tilde{\tm}_K)^{-1}\cl(\mathcal{A})^3\to(\mathcal{A})^3,\quad
(\hat{X},\hat{Y},\hat{Z})\mapsto (X,Y,Z)=
(K^{-1}\hat{X}K,\,\hat{Y}K^{-1}-\hat{Z}K^{-1}\hat{X},\,\hat{Z}).
\end{gather}
One can say that \eqref{amnK}--\eqref{atinvmK} 
are tetrahedron maps in noncommuting variables $X$, $Y$, $Z$, $\hat{X}$, $\hat{Y}$, $\hat{Z}$.
\end{proposition}

\subsection{Tetrahedron maps associated with other matrix refactorisations} 
\label{stmomf}

Let ${\sz},\bm\in\zsp$.
Consider a set $\md$ and maps 
$\mA_{\sz},\mB_{\sz},\mC_{\sz},\mD_{\sz}\cl \md\to\mat_{\sz}(\bK)$.
One has the corresponding map
\begin{gather}
\label{Lsz}
\mL_{\sz}\cl \md\to\mat_{2{\sz}}(\bK),\qqquad
\mL_{\sz}(x)=\begin{pmatrix}
     \mA_{\sz}(x)& \mB_{\sz}(x)  \\
     \mC_{\sz}(x) & \mD_{\sz}(x)
\end{pmatrix},\qqquad x\in\md.
\end{gather}
Now we introduce the maps 
$\mL_{\sz,\bm}^{12},\mL_{\sz,\bm}^{13},\mL_{\sz,\bm}^{23}\cl \md\to\mat_{2{\sz}+\bm}(\bK)$ given by 
\begin{gather}
\label{L12bm}
\mL_{\sz,\bm}^{12}=\begin{pmatrix}
     \mA_{\sz}& \mB_{\sz} &  0 \\
     \mC_{\sz} & \mD_{\sz} & 0 \\
		0 & 0 & \mathbf{1}_{\bm}
\end{pmatrix},\quad
\mL_{\sz,\bm}^{13}=\begin{pmatrix}
     \mA_{\sz} & 0 &  \mB_{\sz} \\
     0 & \mathbf{1}_{\bm} & 0 \\
		\mC_{\sz} & 0 & \mD_{\sz}
\end{pmatrix},\quad
\mL_{\sz,\bm}^{23}=\begin{pmatrix}
    \mathbf{1}_{\bm} &0 &0\\ 
		0& \mA_{\sz} & \mB_{\sz}  \\
    0& \mC_{\sz} & \mD_{\sz} 
\end{pmatrix}.
\end{gather}
Consider the following matrix refactorisation equation
\begin{equation} 
\label{mrfe}
\mL_{\sz,\bm}^{12}(x) \cdot\mL_{\sz,\bm}^{13}(y)\cdot\mL_{\sz,\bm}^{23}(z) = 
\mL_{\sz,\bm}^{23}(\hat{z})\cdot\mL_{\sz,\bm}^{13}(\hat{y})\cdot\mL_{\sz,\bm}^{12}(\hat{x}),\qqquad
x,y,z,\hat{x},\hat{y},\hat{z}\in \md.
\end{equation}

In the case $\bm\neq\sz$ the matrices~\eqref{L12bm} are different from~\eqref{L1213}, 
and, therefore, the matrix refactorisation equation~\eqref{mrfe} 
does not coincide with the matrix local Yang--Baxter equation~\eqref{lybef}.
Below we present new tetrahedron maps arising from equation~\eqref{mrfe}.

Let $\md=\mat_{\sz}(\bK)$. We consider the map
\begin{gather}
\label{Lsz11}
\mL_{\sz}\cl \mat_{\sz}(\bK)\to\mat_{2{\sz}}(\bK),\qqquad
\mL_{\sz}(X)=\begin{pmatrix}
     \mathbf{1}_{\sz}& X  \\
     0 & \mathbf{1}_{\sz}
\end{pmatrix},\qqquad X\in\mat_{\sz}(\bK),
\end{gather}
which is of the form~\eqref{Lsz} with $\mA_{\sz}$, $\mB_{\sz}$, $\mC_{\sz}$, $\mD_{\sz}$ 
given by~\eqref{iabcd}. 
Taking the corresponding matrices~\eqref{L12bm}, 
we see that in the case~\eqref{Lsz11} equation~\eqref{mrfe} reads
\begin{multline} 
\label{mrfe11}
\begin{pmatrix}
     \mathbf{1}_{\sz}& X &  0 \\
     0 & \mathbf{1}_{\sz} & 0 \\
		0 & 0 & \mathbf{1}_{\bm}
\end{pmatrix}
\begin{pmatrix}
     \mathbf{1}_{\sz} & 0 &  Y \\
     0 & \mathbf{1}_{\bm} & 0 \\
		0 & 0 & \mathbf{1}_{\sz}
\end{pmatrix}
\begin{pmatrix}
    \mathbf{1}_{\bm} &0 &0\\ 
		0& \mathbf{1}_{\sz} & Z  \\
    0& 0 & \mathbf{1}_{\sz} 
\end{pmatrix}=\\
=
\begin{pmatrix}
    \mathbf{1}_{\bm} &0 &0\\ 
		0& \mathbf{1}_{\sz} & \hat{Z}  \\
    0& 0 & \mathbf{1}_{\sz} 
\end{pmatrix}
\begin{pmatrix}
     \mathbf{1}_{\sz} & 0 &  \hat{Y} \\
     0 & \mathbf{1}_{\bm} & 0 \\
		0 & 0 & \mathbf{1}_{\sz}
\end{pmatrix}
\begin{pmatrix}
     \mathbf{1}_{\sz}& \hat{X} &  0 \\
     0 & \mathbf{1}_{\sz} & 0 \\
		0 & 0 & \mathbf{1}_{\bm}
\end{pmatrix},\quad
X,Y,Z,\hat{X},\hat{Y},\hat{Z}\in\mat_{\sz}(\bK).
\end{multline}

For given $\sz\times\sz$ matrices $X$, $Y$, $Z$, 
equation~\eqref{mrfe11} does not determine $\sz\times\sz$ 
matrices $\hat{X}$, $\hat{Y}$, $\hat{Z}$ uniquely.
In order to get a map of the form $(X,Y,Z)\mapsto (\hat{X},\hat{Y},\hat{Z})$, 
we need to add some extra equations to~\eqref{mrfe11}.

For example, as shown below, 
one can get a map of the form $(X,Y,Z)\mapsto (\hat{X},\hat{Y},\hat{Z})$,
if one adds the equations $\hat{X}=X$, $\hat{Z}=Z$ to~\eqref{mrfe11}.
Substituting $\hat{X}=X$, $\hat{Z}=Z$ in~\eqref{mrfe11}, we obtain
\begin{multline} 
\label{xyzzvx}
\begin{pmatrix}
     \mathbf{1}_{\sz}& X &  0 \\
     0 & \mathbf{1}_{\sz} & 0 \\
		0 & 0 & \mathbf{1}_{\bm}
\end{pmatrix}
\begin{pmatrix}
     \mathbf{1}_{\sz} & 0 &  Y \\
     0 & \mathbf{1}_{\bm} & 0 \\
		0 & 0 & \mathbf{1}_{\sz}
\end{pmatrix}
\begin{pmatrix}
    \mathbf{1}_{\bm} &0 &0\\ 
		0& \mathbf{1}_{\sz} & Z  \\
    0& 0 & \mathbf{1}_{\sz} 
\end{pmatrix}=\\
=\begin{pmatrix}
    \mathbf{1}_{\bm} &0 &0\\ 
		0& \mathbf{1}_{\sz} & Z  \\
    0& 0 & \mathbf{1}_{\sz} 
\end{pmatrix}
\begin{pmatrix}
     \mathbf{1}_{\sz} & 0 &  \hat{Y} \\
     0 & \mathbf{1}_{\bm} & 0 \\
		0 & 0 & \mathbf{1}_{\sz}
\end{pmatrix}
\begin{pmatrix}
     \mathbf{1}_{\sz}& X &  0 \\
     0 & \mathbf{1}_{\sz} & 0 \\
		0 & 0 & \mathbf{1}_{\bm}
\end{pmatrix},\quad
X,Y,Z,\hat{X},\hat{Y},\hat{Z}\in\mat_{\sz}(\bK).
\end{multline}

Now we need to consider three cases: 
\begin{enumerate}
	\item $\bm\in\{1,\dots,\sz\}$,
	\item $\bm\in\{\sz+1,\dots,2\sz-1\}$,
	\item $\bm\ge 2\sz$.
\end{enumerate}
Below the elements of matrices $X$, $Z$ are denoted by $x_{k,l}$, $z_{k,l}$ 
for $k,l=1,\dots,\sz$, and we use the multiplication of a column by a row
$$
\begin{pmatrix}
x_{1,i}\\
x_{2,i}\\
\vdots\\
x_{\sz,i}
\end{pmatrix}(z_{j,1}\ z_{j,2}\ \dots\ z_{j,\sz})=
\begin{pmatrix}
    x_{1,i}z_{j,1} & x_{1,i}z_{j,2}&\hdots  &x_{1,i}z_{j,n} \\
		x_{2,i}z_{j,1} & x_{2,i}z_{j,2}&\hdots  & x_{2,i}z_{j,n} \\
		\vdots & \vdots & \ddots &\vdots\\
		x_{n,i}z_{j,1} & x_{n,i}z_{j,2}&\hdots  & x_{n,i}z_{j,n}
\end{pmatrix},\qqquad i,j\in\{1,\dots,\sz\}.
$$

Let $\bm\in\{1,\dots,\sz\}$. Then equation~\eqref{xyzzvx} is equivalent to
\begin{gather}
\hat{Y}=Y+\sum_{i=1}^\bm\begin{pmatrix}
x_{1,i}\\
x_{2,i}\\
\vdots\\
x_{\sz,i}
\end{pmatrix}(z_{\sz-\bm+i,1}\ z_{\sz-\bm+i,2}\ \dots\ z_{\sz-\bm+i,\sz}).
\end{gather}
Therefore, we obtain the following map
\begin{gather}
\label{mc1}
\begin{gathered}
\tm_{\sz,\bm}\cl\big(\mat_{\sz}(\bK)\big)^3\to\big(\mat_{\sz}(\bK)\big)^3,
\qqquad\bm\in\{1,\dots,\sz\},\\
(X,Y,Z)\mapsto(\hat{X},\hat{Y},\hat{Z})=
\left(X,\,Y+\sum_{i=1}^\bm\begin{pmatrix}
x_{1,i}\\
x_{2,i}\\
\vdots\\
x_{\sz,i}
\end{pmatrix}(z_{\sz-\bm+i,1}\ z_{\sz-\bm+i,2}\ \dots\ z_{\sz-\bm+i,\sz}),\,Z\right),\\
X,Y,Z\in\mat_{\sz}(\bK).
\end{gathered}
\end{gather}

Now let $\bm\in\{\sz+1,\dots,2\sz-1\}$.
Then \eqref{xyzzvx} is equivalent to
\begin{gather*}
\hat{Y}=Y+\sum_{i=1}^{2\sz-\bm}\begin{pmatrix}
x_{1,\bm-\sz+i}\\
x_{2,\bm-\sz+i}\\
\vdots\\
x_{\sz,\bm-\sz+i}
\end{pmatrix}(z_{i,1}\ z_{i,2}\ \dots\ z_{i,\sz}),
\end{gather*}
and we get the map
\begin{gather}
\label{mc2}
\begin{gathered}
\tilde{\tm}_{\sz,\bm}\cl\big(\mat_{\sz}(\bK)\big)^3\to\big(\mat_{\sz}(\bK)\big)^3,
\qqquad\bm\in\{\sz+1,\dots,2\sz-1\},\\
(X,Y,Z)\mapsto(\hat{X},\hat{Y},\hat{Z})=
\left(X,\,Y+\sum_{i=1}^{2\sz-\bm}\begin{pmatrix}
x_{1,\bm-\sz+i}\\
x_{2,\bm-\sz+i}\\
\vdots\\
x_{\sz,\bm-\sz+i}
\end{pmatrix}(z_{i,1}\ z_{i,2}\ \dots\ z_{i,\sz}),\,Z\right),\\
X,Y,Z\in\mat_{\sz}(\bK).
\end{gathered}
\end{gather}

In the case $\bm\ge 2\sz$ equation~\eqref{xyzzvx} is equivalent to $\hat{Y}=Y$,
and one obtains the identity map $(X,Y,Z)\mapsto(\hat{X},\hat{Y},\hat{Z})=(X,Y,Z)$.
One can verify by a straightforward computation that the maps~\eqref{mc1},~\eqref{mc2}
satisfy the tetrahedron equation.

\subsection{Liouville integrability}
\label{slint} 

In Definition~\ref{dli} we recall the standard notion of Liouville integrability for maps on manifolds
(see, e.g.,~\cite{fordy14,Sokor-Sasha,vesel1991} and references therein).
\begin{definition}
\label{dli}
Let $\dmp\in\zsp$. 
Let $\Mfd$ be a $\dmp$-dimensional manifold 
with (local) coordinates $\xc_1,\dots,\xc_\dmp$.
A (smooth or analytic) map $F\cl \Mfd\to \Mfd$ is said to be \emph{Liouville integrable} 
(or \emph{completely integrable}) if 
one has the following objects on the manifold~$\Mfd$.
\begin{itemize}
	\item A Poisson bracket $\{\,,\,\}$ which is 
	invariant under the map~$F$ and is of constant rank~$2r$ 
	for some positive integer~$r\le\dmp/2$ (i.e., the $\dmp\times\dmp$ matrix with the entries 
	$\{\xc_i,\xc_j\}$ is of constant rank~$2r$).
	The invariance of the bracket means the following.
	For any functions $g$, $h$ on~$\Mfd$ one has 
\begin{gather}
\label{ipbr}
	\{g,h\}\circ F=\{g\circ F,\,h\circ F\}.
\end{gather}
To prove that the bracket is invariant,
it is sufficient to check property~\eqref{ipbr} for $g=\xc_i$, $\,h=\xc_j$, $\,i,j=1,\dots,\dmp$.
	
	In our examples considered below the manifold has a system of coordinates 
	$\xc_1,\dots,\xc_{\dmp}$ such that for any $i,j=1,\dots,\dmp$ 
	the function $\{\xc_i,\xc_j\}$ is constant. 
	Then, in order to prove that the bracket is invariant under~$F$, 
	it is sufficient to show that $\{\xc_i\circ F,\,\xc_j\circ F\}=\{\xc_i,\xc_j\}$ for all $i$, $j$.
	
\item If $2r<\dmp$ then one needs also $\dmp-2r$ functions 
	\begin{gather}
\label{cfs}
C_s,\qqquad s=1,\dots,\dmp-2r,
\end{gather}
	which are invariant under~$F$ (i.e., $C_s\circ F=C_s$) 
	and are Casimir functions (i.e., $\{C_s,g\}=0$ for any function~$g$).
	\item One has $r$ functions 
	\begin{gather}
\label{finv}
I_l,\qqquad l=1,\dots,r,
\end{gather}
	which are invariant under~$F$
	and are in involution with respect to the Poisson bracket (i.e., $\{I_{l_1},I_{l_2}\}=0$ 
	for all $l_1,l_2=1,\dots,r$).
	\item The functions~\eqref{cfs},~\eqref{finv} must be functionally independent.
\end{itemize}	
\end{definition}
	
In this subsection we assume that $\bK$ is either $\bR$ or $\bC$.
Then the set $\big(\mat_{\sz}(\bK)\big)^3$ is a manifold.
Below the elements of matrices $X,Y,Z\in\big(\mat_{\sz}(\bK)\big)^3$ 
are denoted by $x_{i,j}$, $y_{i,j}$, $z_{i,j}$ for $i,j=1,\dots,\sz$.
Then 
\begin{gather}
\label{lcs}
x_{i,j},\quad y_{i,j},\quad z_{i,j},\qqquad
i,j=1,\dots,\sz, 
\end{gather}
can be regarded as coordinates 
on the $3\sz^2$-dimensional manifold $\big(\mat_{\sz}(\bK)\big)^3$.
Clearly, the functions $x_{i,j}$, $z_{i,j}$ are invariant under the maps~\eqref{mc1},~\eqref{mc2}.

\begin{theorem}
\label{emnmc1}
Let $\sz,\bm\in\zsp$ such that $1\le\bm\le\sz/2$.
Then the map~\eqref{mc1} is Liouville integrable.
\end{theorem}
\begin{proof}
On the manifold $\big(\mat_{\sz}(\bK)\big)^3$ we consider the Poisson bracket defined as follows.
The bracket of two coordinates from the list~\eqref{lcs} is nonzero only for
\begin{gather}
\label{pbmc1}
\begin{gathered}
\{y_{p,j},z_{p,j}\}=-\{z_{p,j},y_{p,j}\}=1,\qqquad p=1,\dots,\sz-\bm,\qquad
j=1,\dots,\sz,\\
\{y_{\sz-\bm+q,j},x_{j,\bm+q}\}=-\{x_{j,\bm+q},y_{\sz-\bm+q,j}\}=1,\qqquad q=1,\dots,\bm,\qquad
j=1,\dots,\sz.
\end{gathered}
\end{gather}
This Poisson bracket is of rank~$2\sz^2$.
The $\sz^2$ functions
\begin{gather}
\label{cfmc1}
z_{\sz-\bm+q,j},\quad x_{j,q},\quad x_{j,2\bm+s},\qquad
q=1,\dots,\bm,\quad s=1,\dots,\sz-2\bm,\quad j=1,\dots,\sz,
\end{gather}
are Casimir functions, since they do not appear in~\eqref{pbmc1}.
The rank of the bracket plus 
the number of the Casimir functions~\eqref{cfmc1} equals the dimension of the manifold. 

The $\sz^2$ functions
\begin{gather}
\label{fimc1}
z_{p,j},\qquad x_{j,\bm+q},\qqquad 
p=1,\dots,\sz-\bm,\qquad q=1,\dots,\bm,\qquad j=1,\dots,\sz,
\end{gather}
are in involution with respect to the Poisson bracket, 
as one has $\{x_{i,j},x_{i',j'}\}=\{z_{i,j},z_{i',j'}\}=\{x_{i,j},z_{i',j'}\}=0$ for all $i,j,i',j'=1,\dots,\sz$.
The functions~\eqref{cfmc1},~\eqref{fimc1} are functionally independent and 
are invariant under the map~\eqref{mc1}.

Let us show that the bracket~\eqref{pbmc1} is invariant as well.
In~\eqref{mc1} we have
\begin{gather}
\label{cmc1}
\hat{X}=X,\qqquad
\hat{Y}=Y+\sum_{i=1}^\bm\begin{pmatrix}
x_{1,i}\\
x_{2,i}\\
\vdots\\
x_{\sz,i}
\end{pmatrix}(z_{\sz-\bm+i,1}\ z_{\sz-\bm+i,2}\ \dots\ z_{\sz-\bm+i,\sz}),\qqquad
\hat{Z}=Z.
\end{gather}
The elements of matrices $\hat{X},\hat{Y},\hat{Z}\in\big(\mat_{\sz}(\bK)\big)^3$ 
are denoted by $\hat{x}_{i,j}$, $\hat{y}_{i,j}$, $\hat{z}_{i,j}$ for $i,j=1,\dots,\sz$.
The functions 
\begin{gather}
\label{xizi}
x_{1,i},\ x_{2,i},\ \dots,\ x_{\sz,i},\quad
z_{\sz-\bm+i,1},\ z_{\sz-\bm+i,2},\ \dots,\ z_{\sz-\bm+i,\sz},\qqquad i=1,\dots,\bm,
\end{gather}
which appear in~\eqref{cmc1}, belong to the list of the Casimir functions~\eqref{cfmc1}.
Therefore, formulas~\eqref{pbmc1},~\eqref{cmc1} yield
\begin{gather*}
\{\hat{x}_{i,j},\hat{x}_{i',j'}\}=\{x_{i,j},x_{i',j'}\},\quad
\{\hat{y}_{i,j},\hat{y}_{i',j'}\}=\{y_{i,j},y_{i',j'}\},\quad
\{\hat{z}_{i,j},\hat{z}_{i',j'}\}=\{z_{i,j},z_{i',j'}\},\\
\{\hat{x}_{i,j},\hat{y}_{i',j'}\}=\{x_{i,j},y_{i',j'}\},\quad
\{\hat{x}_{i,j},\hat{z}_{i',j'}\}=\{x_{i,j},z_{i',j'}\},\quad
\{\hat{y}_{i,j},\hat{z}_{i',j'}\}=\{y_{i,j},z_{i',j'}\},\\
i,j,i',j'=1,\dots,\sz,
\end{gather*}
which implies that the Poisson bracket is invariant under the map~\eqref{mc1}.

Therefore, the map~\eqref{mc1} in the case $1\le\bm\le\sz/2$ is Liouville integrable.
\end{proof}

\begin{example}
\label{en2m1}
Let $\sz=2$ and $\bm=1$. Then the map~\eqref{mc1} reads
\begin{gather*}
\begin{gathered}
\tm_{2,1}\cl\big(\mat_{2}(\bK)\big)^3\to\big(\mat_{2}(\bK)\big)^3,\\
(X,Y,Z)\mapsto(\hat{X},\hat{Y},\hat{Z})=
\left(X,\,Y+\begin{pmatrix}
    x_{11}z_{21} & x_{11}z_{22}\\
		x_{21}z_{21} & x_{21}z_{22}
\end{pmatrix},\,Z\right),\\
X=\begin{pmatrix}
    x_{11} & x_{12}\\
		x_{21} & x_{22}
\end{pmatrix},\qqquad 
Y=\begin{pmatrix}
    y_{11} & y_{12}\\
		y_{21} & y_{22}
\end{pmatrix},\qqquad
Z=\begin{pmatrix}
    z_{11} & z_{12}\\
		z_{21} & z_{22}
\end{pmatrix}.
\end{gathered}
\end{gather*}
One has the coordinates 
$x_{ij},y_{ij},z_{ij}$, $i,j\in\{1,2\}$, on the $12$-dimensional 
manifold $\big(\mat_{2}(\bK)\big)^3$.

The Poisson bracket~\eqref{pbmc1} reads
\begin{gather}
\label{pbsc22}
\begin{gathered}
\{y_{11},z_{11}\}=-\{z_{11},y_{11}\}=1,\qqquad
\{y_{12},z_{12}\}=-\{z_{12},y_{12}\}=1,\\
\{y_{21},x_{12}\}=-\{x_{12},y_{21}\}=1,\qqquad
\{y_{22},x_{12}\}=-\{x_{12},y_{22}\}=1,
\end{gathered}
\end{gather}
and is of rank~$8$.
The Casimir functions~\eqref{cfmc1} are $z_{21}$, $z_{22}$, $x_{11}$, $x_{21}$.
The functions~\eqref{fimc1} are $z_{11}$, $z_{12}$, $x_{12}$, $x_{22}$,
and they are in involution with respect to the bracket~\eqref{pbsc22}.
\end{example}

\begin{example}
\label{en5m2}
Now let $\sz=5$ and $\bm=2$. Then the map~\eqref{mc1} is
\begin{gather*}
\begin{gathered}
\tm_{5,2}\cl\big(\mat_{5}(\bK)\big)^3\to\big(\mat_{5}(\bK)\big)^3,\\
(X,Y,Z)\mapsto(\hat{X},\hat{Y},\hat{Z})=
\left(X,\,Y+
\begin{pmatrix}
          x_{11}z_{41}+x_{12}z_{51}& \hdots&x_{11}z_{45}+x_{12}z_{55}\\
          \vdots&\ddots&\vdots\\
          x_{51}z_{41}+x_{52}z_{51}&\hdots&x_{51}z_{45}+x_{52}z_{55}
\end{pmatrix},\,Z\right),\\
X=\begin{pmatrix}
          x_{11}& \hdots&x_{15}\\
          \vdots&\ddots&\vdots\\
          x_{51}&\hdots&x_{55}
\end{pmatrix},\qqquad 
Y=\begin{pmatrix}
          y_{11}& \hdots&y_{15}\\
          \vdots&\ddots&\vdots\\
          y_{51}&\hdots&y_{55}
\end{pmatrix},\qqquad
Z=\begin{pmatrix}
          z_{11}& \hdots&z_{15}\\
          \vdots&\ddots&\vdots\\
          z_{51}&\hdots&z_{55}
\end{pmatrix}.
\end{gathered}
\end{gather*}
One has the coordinates 
$x_{ij},y_{ij},z_{ij}$, $i,j\in\{1,\dots,5\}$, on the $75$-dimensional 
manifold $\big(\mat_{5}(\bK)\big)^3$.

The Poisson bracket~\eqref{pbmc1} reads
\begin{gather}
\label{pb52}
\begin{gathered}
\{y_{1,j},z_{1,j}\}=-\{z_{1,j},y_{1,j}\}=
\{y_{2,j},z_{2,j}\}=-\{z_{2,j},y_{2,j}\}=
\{y_{3,j},z_{3,j}\}=-\{z_{3,j},y_{3,j}\}=1,\\
\{y_{4,j},x_{j,3}\}=-\{x_{j,3},y_{4,j}\}=
\{y_{5,j},x_{j,4}\}=-\{x_{j,4},y_{5,j}\}=1,\qquad j=1,\dots,5,
\end{gathered}
\end{gather}
and is of rank~$50$.
The Casimir functions~\eqref{cfmc1} are $z_{4,j}$, $z_{5,j}$,
$x_{j,1}$, $x_{j,2}$, $x_{j,5}$, $j=1,\dots,5$.
The functions~\eqref{fimc1} are $z_{1,j}$, $z_{2,j}$, $z_{3,j}$, $x_{j,3}$, $x_{j,4}$,
and they are in involution with respect to the bracket~\eqref{pb52}.
\end{example}

\begin{theorem}
\label{emnmc2}
Let $\sz,\bm\in\zsp$ such that $3\sz/2\le\bm\le 2\sz-1$.
Then the map~\eqref{mc2} is Liouville integrable.
\end{theorem}
\begin{proof}
Now on the manifold $\big(\mat_{\sz}(\bK)\big)^3$ we consider the following Poisson bracket.
The bracket of two coordinates from the list~\eqref{lcs} is nonzero only for
\begin{gather}
\label{pbmc2}
\begin{gathered}
\{y_{p,j},z_{2\sz-\bm+p,j}\}=-\{z_{2\sz-\bm+p,j},y_{p,j}\}=1,\qqquad p=1,\dots,\bm-\sz,\qquad
j=1,\dots,\sz,\\
\{y_{\bm-\sz+q,j},x_{j,q}\}=-\{x_{j,q},y_{\sz-\bm+q,j}\}=1,\qqquad q=1,\dots,2\sz-\bm,\qquad
j=1,\dots,\sz.
\end{gathered}
\end{gather}
This bracket is of rank~$2\sz^2$.
The $\sz^2$ functions
\begin{gather}
\label{cfmc2}
z_{q,j},\quad x_{j,2\sz-\bm+p},\qqquad
q=1,\dots,2\sz-\bm,\quad p=1,\dots,\bm-\sz,\quad j=1,\dots,\sz,
\end{gather}
are Casimir functions, as they do not appear in~\eqref{pbmc2}.
The rank of the bracket plus 
the number of the Casimir functions~\eqref{cfmc2} equals the dimension of the manifold. 

The $\sz^2$ functions
\begin{gather}
\label{fimc2}
z_{2\sz-\bm+p,j},\quad x_{j,q},\qqquad
p=1,\dots,\bm-\sz,\qquad
q=1,\dots,2\sz-\bm,\qquad j=1,\dots,\sz,
\end{gather}
are in involution with respect to the bracket, 
since we have $\{x_{i,j},x_{i',j'}\}=\{z_{i,j},z_{i',j'}\}=\{x_{i,j},z_{i',j'}\}=0$ for all $i,j,i',j'=1,\dots,\sz$.
The functions~\eqref{cfmc2},~\eqref{fimc2} are functionally independent and 
are invariant under the map~\eqref{mc2}.

Similarly to the proof of Theorem~\ref{emnmc1}, 
one shows that the bracket~\eqref{pbmc2} is invariant as well.
Hence the map~\eqref{mc2} in the case $3\sz/2\le\bm\le 2\sz-1$ is Liouville integrable.
\end{proof}

\begin{example}
\label{en2m3}
Let $\sz=2$ and $\bm=3$. Then the map~\eqref{mc2} reads
\begin{gather*}
\begin{gathered}
\tilde{\tm}_{2,3}\cl\big(\mat_{2}(\bK)\big)^3\to\big(\mat_{2}(\bK)\big)^3,\\
(X,Y,Z)\mapsto(\hat{X},\hat{Y},\hat{Z})=
\left(X,\,Y+\begin{pmatrix}
    x_{12}z_{11} & x_{12}z_{12}\\
		x_{22}z_{11} & x_{22}z_{12}
\end{pmatrix},\,Z\right),\\
X=\begin{pmatrix}
    x_{11} & x_{12}\\
		x_{21} & x_{22}
\end{pmatrix},\qqquad 
Y=\begin{pmatrix}
    y_{11} & y_{12}\\
		y_{21} & y_{22}
\end{pmatrix},\qqquad
Z=\begin{pmatrix}
    z_{11} & z_{12}\\
		z_{21} & z_{22}
\end{pmatrix}.
\end{gathered}
\end{gather*}
One has the coordinates 
$x_{ij},y_{ij},z_{ij}$, $i,j\in\{1,2\}$, on the $12$-dimensional 
manifold $\big(\mat_{2}(\bK)\big)^3$.

The Poisson bracket~\eqref{pbmc2} reads
\begin{gather}
\label{pbsc23}
\begin{gathered}
\{y_{11},z_{21}\}=-\{z_{21},y_{11}\}=1,\qqquad
\{y_{12},z_{22}\}=-\{z_{22},y_{12}\}=1,\\
\{y_{21},x_{11}\}=-\{x_{11},y_{21}\}=1,\qqquad
\{y_{22},x_{21}\}=-\{x_{21},y_{22}\}=1,
\end{gathered}
\end{gather}
and is of rank~$8$.
The Casimir functions~\eqref{cfmc2} are $z_{11}$, $z_{12}$, $x_{12}$, $x_{22}$.
The functions~\eqref{fimc2} are $z_{21}$, $z_{22}$, $x_{11}$, $x_{21}$,
and they are in involution with respect to the bracket~\eqref{pbsc23}.
\end{example}

\begin{example}
\label{en5m8}
Now let $\sz=5$ and $\bm=8$. Then the map~\eqref{mc2} is
\begin{gather*}
\begin{gathered}
\tilde{\tm}_{5,8}\cl\big(\mat_{5}(\bK)\big)^3\to\big(\mat_{5}(\bK)\big)^3,\\
(X,Y,Z)\mapsto(\hat{X},\hat{Y},\hat{Z})=
\left(X,\,Y+
\begin{pmatrix}
          x_{14}z_{11}+x_{15}z_{21}& \hdots&x_{14}z_{15}+x_{15}z_{25}\\
          \vdots&\ddots&\vdots\\
          x_{54}z_{11}+x_{55}z_{21}&\hdots&x_{55}z_{15}+x_{55}z_{25}
\end{pmatrix},\,Z\right),\\
X=\begin{pmatrix}
          x_{11}& \hdots&x_{15}\\
          \vdots&\ddots&\vdots\\
          x_{51}&\hdots&x_{55}
\end{pmatrix},\qqquad 
Y=\begin{pmatrix}
          y_{11}& \hdots&y_{15}\\
          \vdots&\ddots&\vdots\\
          y_{51}&\hdots&y_{55}
\end{pmatrix},\qqquad
Z=\begin{pmatrix}
          z_{11}& \hdots&z_{15}\\
          \vdots&\ddots&\vdots\\
          z_{51}&\hdots&z_{55}
\end{pmatrix}.
\end{gathered}
\end{gather*}
One has the coordinates 
$x_{ij},y_{ij},z_{ij}$, $i,j\in\{1,\dots,5\}$, on the $75$-dimensional 
manifold $\big(\mat_{5}(\bK)\big)^3$.

The Poisson bracket~\eqref{pbmc2} reads
\begin{gather}
\label{pb58}
\begin{gathered}
\{y_{1,j},z_{3,j}\}=-\{z_{3,j},y_{1,j}\}=
\{y_{2,j},z_{4,j}\}=-\{z_{4,j},y_{2,j}\}=
\{y_{3,j},z_{5,j}\}=-\{z_{5,j},y_{3,j}\}=1,\\
\{y_{4,j},x_{j,1}\}=-\{x_{j,1},y_{4,j}\}=
\{y_{5,j},x_{j,2}\}=-\{x_{j,2},y_{5,j}\}=1,\qquad j=1,\dots,5,
\end{gathered}
\end{gather}
and is of rank~$50$.
The Casimir functions~\eqref{cfmc2} are $z_{1,j}$, $z_{2,j}$,
$x_{j,3}$, $x_{j,4}$, $x_{j,5}$, $j=1,\dots,5$.
The functions~\eqref{fimc2} are $z_{3,j}$, $z_{4,j}$, $z_{5,j}$, $x_{j,1}$, $x_{j,2}$,
and they are in involution with respect to the bracket~\eqref{pb58}.
\end{example}

\section{Tetrahedron maps on groups}
\label{stmgr}

Let $\gmd$ be a group. 
It is known that one has the following Yang--Baxter maps
\begin{subequations}
\label{ybmgr}
\begin{gather}
\label{ym1}
\ym_1\colon \gmd\times \gmd\to \gmd\times \gmd,\qqquad
\ym_1(x,y)=(yxy,y^{-1}),\\
\label{ym2}
\ym_2\colon \gmd\times \gmd\to \gmd\times \gmd,\qqquad
\ym_2(x,y)=(x^{-1},xyx),\\
\label{ym3}
\ym_3\colon\gmd\times \gmd\to \gmd\times \gmd,\qqquad
\ym_3(x,y)=(yx^{-1}y,y),\\
\label{ym4}
\ym_4\colon\gmd\times \gmd\to \gmd\times \gmd,\qqquad
\ym_4(x,y)=(x,xy^{-1}x).
\end{gather}
\end{subequations}
(see, e.g.,~\cite{carter2006,Ito2013} and references therein).
In Theorem~\ref{ptmgr} we present tetrahedron maps of similar type.
\begin{theorem}
\label{ptmgr}
For any group~$\gmd$, one has the following tetrahedron maps
\begin{subequations}
\label{tmgr}
\begin{gather}
\label{yxyz}
\tm_1\colon \gmd^3\to \gmd^3,\qqquad
\tm_1(x,y,z)=(yxy,y^{-1},z^{-1}),\\
\label{tyxyz}
\check{\tm}_1\colon \gmd^3\to \gmd^3,\qqquad
\check{\tm}_1(x,y,z)=(x^{-1},zyz,z^{-1}),\\
\label{x1xy}
\tm_2\colon \gmd^3\to \gmd^3,\qqquad
\tm_2(x,y,z)=(x^{-1},xyx,z^{-1}),\\
\label{tx1xy}
\check{\tm}_2\colon \gmd^3\to \gmd^3,\qqquad
\check{\tm}_2(x,y,z)=(x^{-1},y^{-1},yzy),\\
\label{yx1y}
\tm_3\colon\gmd^3\to \gmd^3,\qqquad
\tm_3(x,y,z)=(yx^{-1}y,y,z^{-1}),\\
\label{tyx1y}
\check{\tm}_3\colon \gmd^3\to \gmd^3,\qqquad
\check{\tm}_3(x,y,z)=(x^{-1},zy^{-1}z,z),\\
\label{xxy1}
\tm_4\colon\gmd^3\to \gmd^3,\qqquad
\tm_4(x,y,z)=(x,xy^{-1}x,z^{-1}),\\
\label{txxy1}
\check{\tm}_4\colon \gmd^3\to \gmd^3,\qqquad
\check{\tm}_4(x,y,z)=(x^{-1},y,yz^{-1}y).
\end{gather}
\end{subequations}
\end{theorem}
\begin{proof}
Applying Proposition~\ref{prop_YB_tetr} to the Yang--Baxter map~\eqref{ym1},
we obtain the following tetrahedron maps
\begin{gather}
\label{syxy}
T=\ym_1\times\id_{\gmd}\colon \gmd^3\to \gmd^3,\qquad
T(x,y,z)=(yxy,y^{-1},z),\\
\label{csyxy}
\check{T}=\id_{\gmd}\times\ym_1\colon 
\gmd^3\to \gmd^3,\qquad
\check{T}(x,y,z)=(x,zyz,z^{-1}).
\end{gather}
Consider
\begin{gather}
\label{sigm}
\sigma\cl\gmd\to\gmd,\qqquad \sigma(x)=x^{-1}.
\end{gather}
For~\eqref{sigm},~\eqref{syxy},~\eqref{csyxy} one has
\begin{gather*}
\sigma\circ\sigma=\id_{\gmd},\qqquad
(\sigma\times\sigma\times\sigma)\circ T\circ(\sigma\times\sigma\times\sigma)=T,
\qqquad 
(\sigma\times\sigma\times\sigma)\circ\check{T}\circ(\sigma\times\sigma\times\sigma)=\check{T}.
\end{gather*}
Therefore, we can apply Proposition~\ref{pist} to the tetrahedron maps~\eqref{syxy},~\eqref{csyxy}, 
which gives the maps~\eqref{yx1y},~\eqref{tyx1y}. 
Indeed, for any $x,y,z\in\gmd$ one has
\begin{gather*}
\big((\sigma\times\id_{\gmd}\times\sigma)\circ T\circ(\id_{\gmd}\times\sigma\times\id_{\gmd})\big)(x,y,z)
=(yx^{-1}y,y,z^{-1})=
\tm_3(x,y,z),\\
\big((\sigma\times\id_{\gmd}\times\sigma)\circ\check{T}
\circ(\id_{\gmd}\times\sigma\times\id_{\gmd})\big)(x,y,z)=
(x^{-1},zy^{-1}z,z)=\check{\tm}_3(x,y,z).
\end{gather*}
Therefore, by Proposition~\ref{pist}, the maps~\eqref{yx1y},~\eqref{tyx1y} satisfy the tetrahedron equation.

Similarly, applying Proposition~\ref{prop_YB_tetr} to the 
Yang--Baxter maps~\eqref{ybmgr}, one obtains the tetrahedron maps
\begin{gather}
\label{ymid}
\ym_i\times\id_{\gmd}\cl\gmd^3\to\gmd^3,\qqquad
\id_{\gmd}\times\ym_i\colon\gmd^3\to\gmd^3,\qqquad i=1,2,3,4.
\end{gather}
Applying Proposition~\ref{pist} to the tetrahedron maps~\eqref{ymid}
with $\sigma$ given by~\eqref{sigm}, one obtains all the maps~\eqref{tmgr}.
Therefore, the maps~\eqref{tmgr} satisfy the tetrahedron equation.
\end{proof}

\begin{remark}
\label{rneyb}
Using~\eqref{sigm} and the Yang--Baxter maps~\eqref{ybmgr} for a group $\gmd$, 
one can rewrite the tetrahedron maps~\eqref{tmgr} as follows
$\tm_i=\ym_i\times\sigma$, $\,\check{\tm}_i=\sigma\times\ym_i$, $\,i=1,2,3,4$.

Note that not for every Yang--Baxter map 
$\ym\cl\gmd\times\gmd\to\gmd\times\gmd$ the maps 
$\ym\times\sigma$ and $\sigma\times\ym$ satisfy the tetrahedron equation.
For example, consider the well-known Yang--Baxter map 
$$
\hat{\ym}\cl\gmd\times\gmd\to\gmd\times\gmd,\qqquad 
\hat{\ym}(x,y)=(x,xyx^{-1}),
$$
which appeared in~\cite{Drin92}. The corresponding maps 
\begin{gather*}
\hat{\ym}\times\sigma\cl\gmd^3\to\gmd^3,\qqquad
(\hat{\ym}\times\sigma)(x,y,z)=(x,xyx^{-1},z^{-1}),\\
\sigma\times\hat{\ym}\cl\gmd^3\to\gmd^3,\qqquad
(\sigma\times\hat{\ym})(x,y,z)=(x^{-1},y,yzy^{-1}),
\end{gather*}
do not satisfy the tetrahedron equation.
\end{remark}

\begin{remark}
For a group $\gmd$, consider
\begin{gather*}
P^{13}\cl \gmd^3\to \gmd^3,\qqquad
P^{13}(a_1,a_2,a_3)=(a_3,a_2,a_1),\qquad a_i\in \gmd.
\end{gather*}
According to Proposition~\ref{pp13}, for any tetrahedron map $T\cl\gmd^3\to\gmd^3$, 
the map $\tilde{T}=P^{13} T P^{13}$ obeys the tetrahedron equation as well.
For the maps~\eqref{tmgr} one has
\begin{gather*}
P^{13} \tm_1 P^{13}=\check{\tm}_2,\qquad
P^{13}\check{\tm}_1 P^{13}={\tm}_2,\qquad
P^{13}{\tm}_2 P^{13}=\check{\tm}_1,\qquad
P^{13} \check{\tm}_2 P^{13}=\tm_1,\\
P^{13} \tm_3 P^{13}=\check{\tm}_4,\qquad
P^{13}\check{\tm}_3 P^{13}={\tm}_4,\qquad
P^{13}{\tm}_4 P^{13}=\check{\tm}_3,\qquad
P^{13} \check{\tm}_4 P^{13}=\tm_3.
\end{gather*}
\end{remark}

\section{Differentials and linear approximations of tetrahedron maps}
\label{sdtm}

In this section, 
when we consider maps of manifolds, 
we assume that they are either smooth, or complex-analytic, or rational,
so that the differential is defined for such a map.

Let $\md$ be a manifold. 
Consider the tangent bundle $\tau\cl T\md\to\md$. 
Then the bundle 
$$
\tau\times\tau\times\tau\cl T\md\times T\md\times T\md\to \md\times \md\times \md
$$ 
can be identified with the tangent bundle of the manifold $\md\times \md\times \md$.
Using this identification and the general procedure 
to define the differential of a map of manifolds, for any map 
$$
{\tm}\cl\md\times \md\times\md\to\md\times\md\times\md
$$
we obtain the differential 
$\dft{\tm}\cl T\md\times T\md\times T\md\to T\md\times T\md\times T\md$.
Proposition~\ref{thdybtet} is proved in~\cite{ikkp21}.
\begin{proposition}[\cite{ikkp21}]
\label{thdybtet}
Let $\md$ be a manifold.
For any tetrahedron map 
${\tm}\cl\md\times \md\times\md\to\md\times\md\times\md$, 
the differential 
$$
\dft{\tm}\cl T\md\times T\md\times T\md\to T\md\times T\md\times T\md
$$ 
is a tetrahedron map of the manifold $T\md\times T\md\times T\md$.
\end{proposition}
A similar result on the differentials of Yang--Baxter maps was used 
in~\cite{BIKRP} without proof and is proved in~\cite{ikkp21}.

Corollary~\ref{ctaaa} is proved in~\cite{ikkp21} as well.
For completeness we present the proof from~\cite{ikkp21}, 
since it clarifies some notions which we use in this paper.
\begin{corollary}[\cite{ikkp21}]
\label{ctaaa}
Consider a manifold $\md$, a tetrahedron map
${\tm}\cl\md\times \md\times\md\to\md\times\md\times\md$, and 
its differential
$$
\dft{\tm}\cl T\md\times T\md\times T\md\to T\md\times T\md\times T\md.
$$

Let $\fp\in\md$ such that ${\tm}\big(({\fp},{\fp},{\fp})\big)=({\fp},{\fp},{\fp})$.
Consider the tangent space $T_{\fp}\md\subset T\md$ at the point $\fp\in\md$.
Then we have
\begin{gather}
\label{dts}
\dft{\tm}(T_{\fp}\md\times T_{\fp}\md\times T_{\fp}\md)\subset
T_{\fp}\md\times T_{\fp}\md\times T_{\fp}\md\subset T\md\times T\md\times T\md,
\end{gather}
and the map 
\begin{gather}
\label{dtaa}
\dft{\tm}\big|_{({\fp},{\fp},{\fp})}\cl
T_{\fp}\md\times T_{\fp}\md\times T_{\fp}\md\to T_{\fp}\md\times T_{\fp}\md\times T_{\fp}\md
\end{gather}
is a linear tetrahedron map. Here $\dft{\tm}\big|_{({\fp},{\fp},{\fp})}$ is 
the restriction of the map $\dft{\tm}$ to $T_{\fp}\md\times T_{\fp}\md\times T_{\fp}\md$. 
\end{corollary}
\begin{proof}
The property ${\tm}\big(({\fp},{\fp},{\fp})\big)=({\fp},{\fp},{\fp})$ and the definition 
of the differential imply~\eqref{dts} and the fact that the map~\eqref{dtaa} is linear.
By Proposition~\ref{thdybtet}, the differential $\dft{\tm}$ is a tetrahedron map.
Therefore, its restriction $\dft{\tm}\big|_{({\fp},{\fp},{\fp})}$ to 
$T_{\fp}\md\times T_{\fp}\md\times T_{\fp}\md$ is a tetrahedron map as well.
\end{proof}

\begin{remark}
\label{rlappr}
As noticed in~\cite{ikkp21}, the definition of the differential implies that 
the linear tetrahedron map $\dft{\tm}\big|_{({\fp},{\fp},{\fp})}$ 
described in Corollary~\ref{ctaaa} can be regarded as a linear approximation of
the nonlinear tetrahedron map~${\tm}$ at the point $({\fp},{\fp},{\fp})\in\md\times\md\times\md$.
Explicit examples of $\dft{\tm}\big|_{({\fp},{\fp},{\fp})}$ are 
presented in Examples~\ref{eeltr},~\ref{edmh},~\ref{enls}.
\end{remark}

Let $\dm\in\zsp$. Let $\md$ be an $\dm$-dimensional manifold 
with (local) coordinates $x_1,\dots,x_\dm$.
Then $\dim T\md =2\dm$, and we have the (local) coordinates $x_1,\dots,x_\dm,X_1,\dots,X_\dm$ 
on the manifold~$T\md$, where $X_i$ corresponds to the differential $\dft x_i$,
which can be regarded as a function on~$T\md$. 
(Thus, the functions $X_1,\dots,X_\dm$ are linear along the fibres of the bundle $T\md\to\md$.)

Following~\cite{ikkp21}, to study maps of the form 
$$
\md\times \md\times\md\to\md\times\md\times\md,\qqquad
T\md\times T\md\times T\md\to T\md\times T\md\times T\md,
$$
we consider 
\begin{itemize}
	\item $6$ copies of the manifold $\md$ with coordinate systems
\begin{gather*}
(x_1,\dots,x_\dm),\quad(y_1,\dots,y_\dm),\quad(z_1,\dots,z_\dm),\quad
(\tilde{x}_1,\dots,\tilde{x}_\dm),\quad(\tilde{y}_1,\dots,\tilde{y}_\dm),\quad(\tilde{z}_1,\dots,\tilde{z}_\dm),
\end{gather*}
\item $6$ copies of the manifold $T\md$ with coordinate systems
\begin{gather*}
(x_1,\dots,x_\dm,X_1,\dots,X_\dm),\qquad(y_1,\dots,y_\dm,Y_1,\dots,Y_\dm),\qquad(z_1,\dots,z_\dm,Z_1,\dots,Z_\dm),\\
(\tilde{x}_1,\dots,\tilde{x}_\dm,\tilde{X}_1,\dots,\tilde{X}_\dm),\qquad
(\tilde{y}_1,\dots,\tilde{y}_\dm,\tilde{Y}_1,\dots,\tilde{Y}_\dm),\qquad
(\tilde{z}_1,\dots,\tilde{z}_\dm,\tilde{Z}_1,\dots,\tilde{Z}_\dm).
\end{gather*}
\end{itemize}
Here, for each $i=1,\dots,\dm$, the functions $X_i$, $Y_i$, $Z_i$, $\tilde{X}_i$, $\tilde{Y}_i$, $\tilde{Z}_i$
correspond to the differentials $\dft x_i$, $\dft y_i$, $\dft z_i$, $\dft\tilde{x}_i$, $\dft\tilde{y}_i$, $\dft\tilde{z}_i$.
Below we use the following  notation
\begin{gather*}
x=(x_1,\dots,x_\dm),\qquad y=(y_1,\dots,y_\dm),\qquad z=(z_1,\dots,z_\dm),\\
X=(X_1,\dots,X_\dm),\qqquad Y=(Y_1,\dots,Y_\dm),\qqquad Z=(Z_1,\dots,Z_\dm),\\
\tilde{x}=(\tilde{x}_1,\dots,\tilde{x}_\dm),\qquad
\tilde{y}=(\tilde{y}_1,\dots,\tilde{y}_\dm),\qquad\tilde{z}=(\tilde{z}_1,\dots,\tilde{z}_\dm),\\
\tilde{X}=(\tilde{X}_1,\dots,\tilde{X}_\dm),\qqquad
\tilde{Y}=(\tilde{Y}_1,\dots,\tilde{Y}_\dm),\qqquad
\tilde{Z}=(\tilde{Z}_1,\dots,\tilde{Z}_\dm).
\end{gather*}

Consider a tetrahedron map
\begin{gather*}
{\tm}\cl\md\times \md\times\md\to\md\times\md\times\md,\qqquad
(x,y,z)\mapsto(\tilde{x},\tilde{y},\tilde{z}),\\
\tilde{x}_i=f_i(x,y,z),\qqquad
\tilde{y}_i=g_i(x,y,z),\qqquad
\tilde{z}_i=h_i(x,y,z),\qqquad i=1,\dots,\dm.
\end{gather*}
Its differential is the following tetrahedron map 
\begin{gather*}
\dft{\tm}\cl T\md\times T\md\times T\md\to T\md\times T\md\times T\md,\qqquad
(x,X,y,Y,z,Z)\mapsto(\tilde{x},\tilde{X},\tilde{y},\tilde{Y},\tilde{z},\tilde{Z}),\\
\tilde{x}_i=f_i(x,y,z),\qqquad
\tilde{y}_i=g_i(x,y,z),\qqquad
\tilde{z}_i=h_i(x,y,z),\qqquad i=1,\dots,\dm,\\
\tilde{X}_i=\sum_{j=1}^\dm\Big(\frac{\partial f_i(x,y,z)}{\partial x_j}X_j
+\frac{\partial f_i(x,y,z)}{\partial y_j}Y_j+\frac{\partial f_i(x,y,z)}{\partial z_j}Z_j\Big),\\
\tilde{Y}_i=\sum_{j=1}^\dm\Big(\frac{\partial g_i(x,y,z)}{\partial x_j}X_j
+\frac{\partial g_i(x,y,z)}{\partial y_j}Y_j+\frac{\partial g_i(x,y,z)}{\partial z_j}Z_j\Big),\\
\tilde{Z}_i=\sum_{j=1}^\dm\Big(\frac{\partial h_i(x,y,z)}{\partial x_j}X_j
+\frac{\partial h_i(x,y,z)}{\partial y_j}Y_j+\frac{\partial h_i(x,y,z)}{\partial z_j}Z_j\Big).
\end{gather*}

Example~\ref{eeltr} is taken from~\cite{ikkp21}.
\begin{example}
\label{eeltr}
Let $\dm=\dim\md=1$.
Consider the well-known electric network transformation
\begin{gather}
\label{elnt1}
{\tm}\cl\md\times \md\times\md\to\md\times\md\times\md,\qqquad
(x,y,z)\mapsto (\tilde{x},\tilde{y},\tilde{z}),\\
\label{elnt2}
\tilde{x}=\frac{x y}{x + z + x y z},\qqquad
\tilde{y}=x + z + x y z,\qqquad
\tilde{z}=\frac{y z}{x y z+x+z},
\end{gather}
which is a tetrahedron map~\cite{Sergeev,Kashaev-Sergeev}.
As shown in~\cite{ikkp21}, its differential is the following tetrahedron map 
\begin{gather}
\label{difelt}
\dft{\tm}\cl T\md\times T\md\times T\md\to T\md\times T\md\times T\md,\qqquad
(x,X,y,Y,z,Z)\mapsto(\tilde{x},\tilde{X},\tilde{y},\tilde{Y},\tilde{z},\tilde{Z}),\\
\notag
\tilde{x}=\frac{x y}{x + z + x y z},\qqquad
\tilde{y}=x + z + x y z,\qqquad
\tilde{z}=\frac{y z}{x y z+x+z},\\
\label{delnt1}
\tilde{X}=\frac{- x y (1 + x y)Z +  y zX +  x (x + z)Y}{(x y z+x+z)^2},\qqquad
\tilde{Y}=X + Z + x y Z+ x zY + y zX,\\
\label{delnt2}
\tilde{Z}=\frac{- y z (y z+1)X+z (x+z)Y+x yZ}{(x y z+x+z)^2}.
\end{gather}

We assume that $x,y,z,\tilde{x},\tilde{y},\tilde{z}$ take values in~$\mathbb{C}$, 
so $\md$ is a complex manifold.
Consider $\mathrm{i}\in\mathbb{C}$ satisfying $\mathrm{i}^2=-1$.
Formulas~\eqref{elnt1},~\eqref{elnt2} imply 
${\tm}\big((\mathrm{i},\mathrm{i},\mathrm{i})\big)=
(\mathrm{i},\mathrm{i},\mathrm{i})$.

Let $\fp=\mathrm{i}$.
The coordinate system on~$\md$ gives the isomorphism $T_{\fp}\md\cong\mathbb{C}$.
By Corollary~\ref{ctaaa}, we obtain the linear tetrahedron map 
$\dft{\tm}\big|_{(\mathrm{i},\mathrm{i},\mathrm{i})}\cl 
\mathbb{C}^3\to\mathbb{C}^3$. To compute it, 
we substitute $x=y=z=\mathrm{i}$ in~\eqref{delnt1},~\eqref{delnt2} and derive 
the linear map
\begin{gather}
\label{xtx}
\begin{pmatrix}
     X  \\
     Y  \\
     Z
\end{pmatrix}\mapsto
\begin{pmatrix}
     \tilde{X}  \\
     \tilde{Y}  \\
     \tilde{Z}
\end{pmatrix}=
\begin{pmatrix}
     X+2Y  \\
     -Y  \\
     2Y+Z
\end{pmatrix}.
\end{gather}

The linear tetrahedron map~\eqref{xtx} is a linear approximation 
of the nonlinear map~\eqref{elnt1},~\eqref{elnt2} at the point $(\mathrm{i},\mathrm{i},\mathrm{i})$
in the following sense. We have
\begin{gather*}
{\tm}\big((\mathrm{i}+\varepsilon X,\,
\mathrm{i}+\varepsilon Y,\,\mathrm{i}+\varepsilon Z)\big)=
\big(\mathrm{i}+\varepsilon(X+2Y)+\mathcal{O}(\varepsilon^2),\,
\mathrm{i}-\varepsilon Y+\mathcal{O}(\varepsilon^2),\,
\mathrm{i}+\varepsilon(2Y+Z)+\mathcal{O}(\varepsilon^2)\big).
\end{gather*}
\end{example}

\begin{remark}
\label{rlth}
Results of Hietarinta~\cite{Hietarinta97} imply that for any field $\bK$ 
and any $\ah,\bh,\ch\in\bK$ the linear map
\begin{gather}
\label{hietm}
T\cl\bK^3\to\bK^3,\qqquad
\begin{pmatrix}
     X  \\
     Y  \\
     Z
\end{pmatrix}\mapsto
\begin{pmatrix}
     \tilde{X}  \\
     \tilde{Y}  \\
     \tilde{Z}
\end{pmatrix}=
\begin{pmatrix}
     \ah X+(1-\ah\bh)Y  \\
     -\bh Y  \\
     (1-\bh\ch)Y+\ch Z
\end{pmatrix},
\end{gather}
is tetrahedron map. The map~\eqref{xtx}
is of the form~\eqref{hietm} for $\ah=1$, $\bh=-1$, $\ch=1$.
\end{remark}

\begin{proposition}
\label{pfpt}
Consider a set $\md$ and a map $\tm\cl \md^3\to \md^3$.
Suppose that for some ${\sz}\in\zsp$ there are maps 
$\mA_{\sz},\mB_{\sz},\mC_{\sz},\mD_{\sz}\cl \md\to\mat_{\sz}(\bK)$
such that for the corresponding map 
\begin{gather}
\label{mlab}
\mL=\begin{pmatrix}
     \mA_{\sz}& \mB_{\sz}  \\
     \mC_{\sz} & \mD_{\sz}
\end{pmatrix}\cl \md\to\mat_{2{\sz}}(\bK)
\end{gather}
we have the following property\textup{:}

\noindent the local Yang--Baxter equation 
\begin{equation} 
\label{lybe}
\mL^{12} (x) \cdot\mL^{13} (y)\cdot\mL^{23}(z) = 
\mL^{23}(\hat{z})\cdot\mL^{13}(\hat{y})\cdot\mL^{12}(\hat{x}),\qqquad
x,y,z,\hat{x},\hat{y},\hat{z}\in \md,
\end{equation}
is equivalent to the relation 
\begin{gather}
\label{xyzt}
(\hat{x},\hat{y},\hat{z})=\tm\big((x,y,z)\big),\qqquad
x,y,z,\hat{x},\hat{y},\hat{z}\in \md.
\end{gather}
Here $\mL^{12},\mL^{13},\mL^{23}\cl \md\to\mat_{3{\sz}}(\bK)$ are given by~\eqref{L1213}.

Then a point $\fp\in \md$ satisfies 
\begin{gather} 
\label{tppp}
\tm\big((\fp,\fp,\fp)\big)=(\fp,\fp,\fp)
\end{gather}
if and only if the linear map
\begin{gather}
\label{lybl}
\mL(\fp)=\begin{pmatrix}
     \mA_{\sz}(\fp)& \mB_{\sz}(\fp)  \\
     \mC_{\sz}(\fp) & \mD_{\sz}(\fp)
\end{pmatrix}\cl\bK^{2{\sz}}\to\bK^{2{\sz}}
\end{gather}
obeys the Yang--Baxter equation, that is,
\begin{gather} 
\label{mlfp}
\mL^{12} (\fp) \cdot\mL^{13} (\fp)\cdot\mL^{23}(\fp) = 
\mL^{23}(\fp)\cdot\mL^{13}(\fp)\cdot\mL^{12}(\fp).
\end{gather}
\end{proposition}
\begin{proof}
By assumption, \eqref{lybe} is equivalent to~\eqref{xyzt}.
Substituting $x=y=z=\fp$ and $\hat{x}=\hat{y}=\hat{z}=\fp$ 
in~\eqref{lybe},~\eqref{xyzt}, we derive that 
\eqref{mlfp} is equivalent to~\eqref{tppp}.
\end{proof}

\begin{remark}
\label{rlybinv}
In Proposition~\ref{pfpt} equation~\eqref{lybe} can be replaced by the equation
\begin{equation} 
\label{lybinv}
\mL^{23} (z) \cdot\mL^{13} (y)\cdot\mL^{12}(x) = 
\mL^{12}(\hat{x})\cdot\mL^{13}(\hat{y})\cdot\mL^{23}(\hat{z}),
\qqquad
x,y,z,\hat{x},\hat{y},\hat{z}\in \md.
\end{equation}
That is, if equation~\eqref{lybinv} is equivalent to relation~\eqref{xyzt}
then the conclusion of Proposition~\ref{pfpt} remains valid.
\end{remark}

\begin{remark}
\label{rratm}
Let $\md$ be an algebraic variety and $\tm\cl\md^3\to\md^3$ 
be a rational map defined on an open dense subset of~$\md^3$.
Proposition~\ref{pfpt} remains valid in this situation, 
provided that we restrict considerations to the open dense subset of~$\md^3$
on which the map~$\tm$ is defined.

When we consider the expressions 
$\tm\big((x,y,z)\big)$ and $\tm\big((\fp,\fp,\fp)\big)$ 
for such a rational map $\tm\cl\md^3\to\md^3$,
we always assume that $x,y,z,\fp\in\md$ are such that 
$\tm$ is defined at~$(x,y,z)\in\md^3$ and at~$(\fp,\fp,\fp)\in\md^3$.
\end{remark}

\begin{remark}
\label{ryb22}
Results of~\cite{Hietarinta97} imply the following classification 
of linear Yang--Baxter maps $Y\cl\bK\times\bK\to\bK\times\bK$.

A linear map $Y\cl\bK\times\bK\to\bK\times\bK$ 
can be identified with a $2\times 2$ matrix $Y=\begin{pmatrix}
     a&b  \\
     c&d
\end{pmatrix}$, $a,b,c,d\in\bK$.
Then the maps $Y^{12},Y^{13},Y^{23}\cl\bK^3\to\bK^3$ 
are given by the matrices
$$
Y^{12}=\begin{pmatrix}
     a&b&0  \\
     c&d&0 \\
		0&0&1
\end{pmatrix},\qquad Y^{13}=\begin{pmatrix}
     a&0&b  \\
     0&1&0 \\
		c&0&d
\end{pmatrix},\qquad Y^{23}=\begin{pmatrix}
     1&0&0  \\
     0&a&b \\
		0&c&d
\end{pmatrix}.
$$
The Yang--Baxter equation~\eqref{eq_YB} reads 
\begin{gather}
\label{mybeq}
\begin{pmatrix}
     a&b&0  \\
     c&d&0 \\
		0&0&1
\end{pmatrix}\begin{pmatrix}
     a&0&b  \\
     0&1&0 \\
		c&0&d
\end{pmatrix}
\begin{pmatrix}
     1&0&0  \\
     0&a&b \\
		0&c&d
\end{pmatrix}=
\begin{pmatrix}
     1&0&0  \\
     0&a&b \\
		0&c&d
\end{pmatrix}
\begin{pmatrix}
     a&0&b  \\
     0&1&0 \\
		c&0&d
\end{pmatrix}
\begin{pmatrix}
     a&b&0  \\
     c&d&0 \\
		0&0&1
\end{pmatrix}.
\end{gather} 
As shown in~\cite{Hietarinta97}, a matrix $\begin{pmatrix}
     a&b  \\
     c&d
\end{pmatrix}$ obeys~\eqref{mybeq} if and only if one has
 \begin{gather}
\label{abcdrel}
abc=0,\qquad bcd=0,\qquad bc(b-c)=0,\qquad b(ad+b-1)=0,\qquad c(ad+c-1)=0.
\end{gather} 
As noticed in~\cite{Hietarinta97}, 
it is easy to show that a matrix $\begin{pmatrix}
     a&b  \\
     c&d
\end{pmatrix}$
satisfies~\eqref{abcdrel} if and only if this matrix
belongs to one of the following classes
\begin{gather}
\label{clyb}
\begin{pmatrix}
     \ah& 0\\
     0&\ddh
\end{pmatrix},\qquad
\begin{pmatrix}
     \ah&1-\ah\ddh\\
     0&\ddh
\end{pmatrix},\qquad
\begin{pmatrix}
     \ah& 0\\
     1-\ah\ddh &\ddh
\end{pmatrix},\qquad
\begin{pmatrix}
     0&1\\
     1&0
\end{pmatrix},\qquad
\ah,\ddh\in\bK.
\end{gather}
\end{remark}

\begin{remark}
\label{rlaxr}
It is known that the local Yang--Baxter equation~\cite{Nijhoff-2} can be regarded 
as a ``Lax equation'' or ``Lax system'' for the tetrahedron equation 
(see, e.g.,~\cite{DMH15} and references therein).
By analogy with the notions of 
Lax matrices, Lax representations, and strong Lax matrices 
for Yang-Baxter maps~\cite{Veselov2,Kouloukas2},
one can use the following ``Lax terminology'' for tetrahedron maps.

Suppose that we have a tetrahedron map $\tm\cl\md^3\to\md^3$
and a matrix-function $\mL$ of the form~\eqref{mlab}
such that relation~\eqref{xyzt} implies the local Yang--Baxter equation~\eqref{lybe}. 
Then $\mL$ can be called a Lax matrix (or Lax representation) for the map~$\tm$.
If $\mL$ is such that relation~\eqref{xyzt} is equivalent to equation~\eqref{lybe} 
then $\mL$ can be called a strong Lax matrix (or strong Lax representation) for~$\tm$.
Using this terminology, 
one can say that in Proposition~\ref{pfpt} we consider a map~$\tm$ 
with a strong Lax representation.

Note that the paper~\cite{KR} uses the term ``Lax representation'' 
in a slightly different sense. 
Namely, if a tetrahedron map $\tm\cl\md^3\to\md^3$ and 
a matrix-function $\mL$ of the form~\eqref{mlab} are such that 
relation~\eqref{xyzt} is equivalent to equation~\eqref{lybinv} then the paper~\cite{KR}
says that equation~\eqref{lybinv} is a Lax representation for the map~$\tm$.
\end{remark}

\begin{example}
\label{edmh}
Dimakis and M\"uller-Hoissen~\cite{Dimakis} 
constructed the rational tetrahedron map
\begin{gather}
\label{dmt}
\tm\cl \bC^2\times \bC^2\times \bC^2\to \bC^2\times \bC^2\times \bC^2,\quad
\big((x_1,x_2),(y_1,y_2),(z_1,z_2)\big)\mapsto 
\big((\tilde{x}_1,\tilde{x}_2),(\tilde{y}_1,\tilde{y}_2),(\tilde{z}_1,\tilde{z}_2)\big),\\
\notag
\tilde{x}_1=y_1 \mathsf{C},\qqquad
\tilde{x}_2=\Big(y_1 - \frac{\mathsf{A}}{x_1}\Big) \mathsf{C},\qqquad
\tilde{y}_1=\frac{x_1}{\mathsf{C}},\qqquad 
\tilde{y}_2=1 - \mathsf{B},\\ 
\notag
\tilde{z}_1=\frac{z_1 y_1(x_1-x_2)}{\mathsf{A}},\qqquad 
\tilde{z}_2=1 - \frac{(1-y_2) (1-z_2)}{\mathsf{B}},\\
\notag
\mathsf{A} = y_2 z_1 x_1-y_2 x_1-z_1 x_2+x_1 y_1,\qqquad 
\mathsf{B} = y_2 z_2 x_1-y_2 x_1-z_2 x_2+1,\\
\notag
\mathsf{C} = \frac{\mathsf{A} \mathsf{B}-\mathsf{A} (1-y_2) (1-z_2) x_1-\mathsf{B} z_1(x_1-x_2)}{\mathsf{A} \mathsf{B}-\mathsf{A} (1-y_2) (1-z_2)-\mathsf{B} z_1 y_1(x_1-x_2)}.
\end{gather}

Let $\md=\bC^2$ and ${\sz}=1$.
Following~\cite[Section~10]{Dimakis}, we consider the $2\times 2$ matrix-function 
\begin{gather}
\notag
\mL\cl\bC^2\to\mat_2(\bC),\qqquad
\mL(x)=\begin{pmatrix}
       x_1 & x_2 \\
       1-x_1 & 1-x_2
\end{pmatrix},\qqquad x=(x_1,x_2)\in\bC^2.
\end{gather}
As shown in~\cite[Section~10]{Dimakis}, 
for this $\mL$ and $\md=\bC^2$ equation~\eqref{lybe} 
is equivalent to relation~\eqref{xyzt} for the map~\eqref{dmt}.
(This equivalence holds on the open dense subset of 
$\bC^2\times \bC^2\times \bC^2$ on which the rational map~\eqref{dmt} is defined.)
Hence $\mL$ is a strong Lax representation for~\eqref{dmt} in the sense of Remark~\ref{rlaxr}.

Using Proposition~\ref{pfpt}, 
we are going to find points $\fp\in \md$ satisfying $\tm\big((\fp,\fp,\fp)\big)=(\fp,\fp,\fp)$.
According to Proposition~\ref{pfpt} and Remarks~\ref{rratm},~\ref{ryb22},
it is sufficient to find points $\fp=(x_1,x_2)\in\bC^2$ 
satisfying the following conditions:
\begin{gather}
\label{cond1}
\text{the $2\times 2$ matrix $\mL(\fp)$ belongs to one of the classes~\eqref{clyb}},\\
\label{cond2}
\text{the rational map~\eqref{dmt} is defined at the point $(\fp,\fp,\fp)$.}
\end{gather}
Condition~\eqref{cond2} means that 
the denominators of the rational functions in the formula 
for the map~\eqref{dmt} must not vanish at the point $(\fp,\fp,\fp)$.

Conditions~\eqref{cond1},~\eqref{cond2} are valid in the following cases
\begin{gather}
\label{case1}
\fp=(\ct,0),\qqquad\mL(\fp)=\begin{pmatrix}
       \ct & 0 \\
       1-\ct & 1
\end{pmatrix},
\qqquad\ct\in\bC,\qquad \ct\neq 0,\\
\label{case2}
\fp=(1,\qt),\qqquad\mL(\fp)=\begin{pmatrix}
       1 & \qt \\
       0 & 1-\qt
\end{pmatrix},
\qqquad\qt\in\bC,\qquad \qt\neq 1.
\end{gather}
Therefore, in the cases~\eqref{case1},~\eqref{case2} 
we have $\tm\big((\fp,\fp,\fp)\big)=(\fp,\fp,\fp)$.
Since~$\md=\bC^2$, one has $\tm_{\fp}\md\cong\mathbb{C}^2$.
Hence, by Corollary~\ref{ctaaa}, we obtain the linear tetrahedron map
$$
\dft \tm\big|_{(\fp,\fp,\fp)}\cl\mathbb{C}^2\times\mathbb{C}^2\times\mathbb{C}^2
\to\mathbb{C}^2\times\mathbb{C}^2\times\mathbb{C}^2.
$$

Computing $\dft \tm$ and $\dft \tm\big|_{(\fp,\fp,\fp)}$ in the case~\eqref{case1}, 
one derives that $\dft \tm\big|_{(\fp,\fp,\fp)}$ is given by
\begin{gather}
\label{lmct}
\begin{gathered}
\dft \tm\big|_{(\fp,\fp,\fp)}\cl
\mathbb{C}^2\times\mathbb{C}^2\times\mathbb{C}^2
\to\mathbb{C}^2\times\mathbb{C}^2\times\mathbb{C}^2,\qqquad
\fp=(\ct,0),\qquad\ct\in\bC,\qquad\ct\neq 0,\\
\big((X_1,X_2),(Y_1,Y_2),(Z_1,Z_2)\big)\mapsto 
\big((\tilde{X}_1,\tilde{X}_2),(\tilde{Y}_1,\tilde{Y}_2),(\tilde{Y}_1,\tilde{Y}_2)\big),
\end{gathered}\\
\notag
\tilde{X}_1={X_1}+\frac{(\ct-1)}{\ct} {Y_1}+\frac{(\ct+1) (\ct-1)^2}{\ct} {Y_2}-\frac{(\ct-1)}{\ct} {Z_1}+(1-\ct) {Z_2},\\
\notag
\tilde{X}_2={X_2}+(1-\ct) {Y_2},\\
\notag
\tilde{Y}_1=\frac{1}{\ct}{Y_1}-\frac{(\ct+1) (\ct-1)^2}{\ct} {Y_2}+\frac{(\ct-1)}{\ct}{Z_1}+(\ct-1) {Z_2},\\
\notag
\tilde{Y}_2=\ct {Y_2},\qqquad
\tilde{Z}_1=(1-\ct) {Y_2}+{Z_1},\qqquad
\tilde{Z}_2=(1-\ct) {Y_2}+{Z_2}.
\end{gather}

In the case~\eqref{case2} one obtains
\begin{gather}
\label{lmqt}
\begin{gathered}
\dft \tm\big|_{(\fp,\fp,\fp)}\cl\mathbb{C}^2\times\mathbb{C}^2\times\mathbb{C}^2
\to\mathbb{C}^2\times\mathbb{C}^2\times\mathbb{C}^2,\qqquad
\fp=(1,\qt),\qquad\qt\in\bC,\qquad\qt\neq 1,\\
\big((X_1,X_2),(Y_1,Y_2),(Z_1,Z_2)\big)\mapsto 
\big((\tilde{X}_1,\tilde{X}_2),(\tilde{Y}_1,\tilde{Y}_2),(\tilde{Y}_1,\tilde{Y}_2)\big),
\end{gathered}\\
\notag
\tilde{X}_1={X_1}+\frac{\qt}{\qt-1} {Y_1},\qqquad
\tilde{X}_2={X_2}+\frac{\qt}{\qt-1} {Y_1},\qqquad
\tilde{Y}_1=\frac{1}{1-\qt}{Y_1},\\
\notag
\tilde{Y}_2=(1-\qt) \qt {X_1}+\qt {X_2}+(1-\qt) {Y_2},\\
\notag
\tilde{Z}_1=\frac{\qt}{\qt-1} {Y_1}+{Z_1},\qqquad
\tilde{Z}_2=(\qt-1) \qt {X_1}-\qt {X_2}+\qt {Y_2}+{Z_2}.
\end{gather}
\end{example}

\begin{example}
\label{enls}
One of us~\cite{KR} constructed the rational tetrahedron map
\begin{gather}
\label{nlst}
\begin{gathered}
\tm\cl \bC^3\times \bC^3\times \bC^3\to \bC^3\times \bC^3\times \bC^3,\\
\big((x_1,x_2,x_3),(y_1,y_2,y_3),(z_1,z_2,z_3)\big)\mapsto 
\big((\tilde{x}_1,\tilde{x}_2,\tilde{x}_3),(\tilde{y}_1,\tilde{y}_2,\tilde{y}_3),
(\tilde{z}_1,\tilde{z}_2,\tilde{z}_3)\big),
\end{gathered}\\
\notag
\tilde{x}_1=\frac{y_3 x_1-y_1z_2}{z_3},\qqquad
\tilde{x}_2=
\frac{x_3z_3(x_2(z_3+z_1z_2)+y_2z_1(x_3+x_1x_2))}{x_3y_3z_3-(x_2z_2+y_2(x_3+x_1x_2))(y_3x_1z_1-y_1(z_3+z_1z_2))},\\
\notag
\tilde{x}_3=x_3,\qquad
\tilde{y}_1=\frac{y_1(z_3+z_1z_2)-y_3x_1z_1}{x_3z_3},\qquad
\tilde{y}_2=x_2z_2+y_2(x_3+x_1x_2),\qquad
\tilde{y}_3=y_3,\\
\notag
\tilde{z}_1=
\frac{z_3(y_3z_1(x_3+x_1x_2)-x_2y_1(z_3+z_1z_2))}{x_3y_3z_3-(x_2z_2+y_2(x_3+x_1x_2))(y_3x_1z_1-y_1(z_3+z_1z_2))},\\
\notag
\tilde{z}_2=x_1y_2+z_2,\qqquad
\tilde{z}_3=z_3.
\end{gather}
(These formulas appear in~\cite{KR} in different notation: $x_3,y_3,z_3$ are denoted by $a,b,c$ in~\cite{KR}.)

Let $\md=\bC^3$ and ${\sz}=1$.
Following~\cite[Section~4]{KR}, we consider the $2\times 2$ matrix-function 
\begin{gather}
\notag
\mL\cl\bC^3\to\mat_2(\bC),\qqquad
\mL(x)=\begin{pmatrix}
        x_3+x_1x_2 & x_1\\
         x_2 & 1
				\end{pmatrix},\qqquad x=(x_1,x_2,x_3)\in\bC^3.
\end{gather}
For this~$\mL$ and $\md=\bC^3$ equation~\eqref{lybinv} 
holds as a consequence of relation~\eqref{xyzt} for the map~\eqref{nlst}, 
but, in this case, \eqref{lybinv} is not equivalent to~\eqref{xyzt}.
Hence $\mL$ is a Lax representation for the map inverse to~\eqref{nlst},
but this Lax representation is not strong.

Therefore, in this case, Remark~\ref{rlybinv} and Proposition~\ref{pfpt} 
cannot be applied immediately. However, the idea of Proposition~\ref{pfpt} 
can be applied here with the following slight modifications.

Below we use the notation
\begin{gather*}
x=(x_1,x_2,x_3),\,\ y=(y_1,y_2,y_3),\,\  z=(z_1,z_2,z_3),\,\ 
\tilde{x}=(\tilde{x}_1,\tilde{x}_2,\tilde{x}_3),\,\ 
\tilde{y}=(\tilde{y}_1,\tilde{y}_2,\tilde{y}_3),\,\ 
\tilde{z}=(\tilde{z}_1,\tilde{z}_2,\tilde{z}_3),\\
x_j,y_j,z_j,\tilde{x}_j,\tilde{y}_j,\tilde{z}_j\in\bC,\qqquad
j=1,2,3.
\end{gather*}
It is easy to check that, for the map~\eqref{nlst}, the relation
\begin{gather*}
(\tilde{x},\tilde{y},\tilde{z})=\tm\big((x,y,z)\big),\qqquad
x,y,z,\tilde{x},\tilde{y},\tilde{z}\in\bC^3,
\end{gather*}
is equivalent to the system
\begin{gather*} 
\mL^{23} (z) \cdot\mL^{13} (y)\cdot\mL^{12}(x) = 
\mL^{12}(\tilde{x})\cdot\mL^{13}(\tilde{y})\cdot\mL^{23}(\tilde{z}),
\qqquad
x,y,z,\tilde{x},\tilde{y},\tilde{z}\in\bC^3,\\
\tilde{x}_3=x_3,\qqquad
\tilde{y}_3=y_3,\qqquad
\tilde{z}_3=z_3.
\end{gather*}
(This equivalence holds on the open dense subset of 
$\bC^3\times \bC^3\times \bC^3$ on which the rational map~\eqref{nlst} is defined.)

Therefore, if a point $\fp=(x_1,x_2,x_3)\in\bC^3$ satisfies the following conditions:
\begin{gather}
\label{ncond1}
\text{the $2\times 2$ matrix $\mL(\fp)$ belongs to one of the classes~\eqref{clyb}},\\
\label{ncond2}
\text{the rational map~\eqref{nlst} is defined at the point $(\fp,\fp,\fp)$,}
\end{gather}
then we have $\mL^{23}(\fp)\cdot\mL^{13}(\fp)\cdot\mL^{12}(\fp)=
\mL^{12}(\fp) \cdot\mL^{13} (\fp)\cdot\mL^{23}(\fp)$ and 
$\tm\big((\fp,\fp,\fp)\big)=(\fp,\fp,\fp)$.

Condition~\eqref{ncond2} means that 
the denominators of the rational functions in the formula 
for the map~\eqref{nlst} must not vanish at the point $(\fp,\fp,\fp)$.

Conditions~\eqref{ncond1},~\eqref{ncond2} are valid in the following cases
\begin{gather}
\label{ncase1}
\fp=(0,1-\qt_1,\qt_1),\qqquad\mL(\fp)=\begin{pmatrix}
       \qt_1 & 0 \\
       1-\qt_1 & 1
\end{pmatrix},
\qqquad\qt_1\in\bC,\qquad\qt_1\neq 0,\\
\label{ncase2}
\fp=(1-\qt_2,0,\qt_2),\qqquad\mL(\fp)=\begin{pmatrix}
       \qt_2 & 1-\qt_2 \\
       0 & 1
\end{pmatrix},
\qqquad\qt_2\in\bC,\qquad\qt_2\neq 0,\\
\label{ncase3}
\fp=(0,0,\qt_3),\qqquad\mL(\fp)=\begin{pmatrix}
       \qt_3 & 0 \\
       0 & 1
\end{pmatrix},
\qqquad\qt_3\in\bC,\qquad\qt_3\neq 0.
\end{gather}
Therefore, in the cases~\eqref{ncase1},~\eqref{ncase2},~\eqref{ncase3} 
we have $\tm\big((\fp,\fp,\fp)\big)=(\fp,\fp,\fp)$.
Since~$\md=\bC^3$, one has $\tm_{\fp}\md\cong\mathbb{C}^3$.
Hence, by Corollary~\ref{ctaaa}, we obtain the linear tetrahedron map
$$
\dft \tm\big|_{(\fp,\fp,\fp)}\cl\mathbb{C}^3\times\mathbb{C}^3\times\mathbb{C}^3
\to\mathbb{C}^3\times\mathbb{C}^3\times\mathbb{C}^3.
$$

Computing $\dft \tm$ and $\dft \tm\big|_{(\fp,\fp,\fp)}$ in the case~\eqref{ncase1},
one derives that $\dft \tm\big|_{(\fp,\fp,\fp)}$ is given by
\begin{gather}
\label{lmqt1}
\begin{gathered}
\dft \tm\big|_{(\fp,\fp,\fp)}\cl\mathbb{C}^3\times\mathbb{C}^3\times\mathbb{C}^3
\to\mathbb{C}^3\times\mathbb{C}^3\times\mathbb{C}^3,\qquad
\fp=(0,1-\qt_1,\qt_1),\qquad\qt_1\in\bC,\qquad\qt_1\neq 0,\\
\big((X_1,X_2,X_3),(Y_1,Y_2,Y_3),(Z_1,Z_2,Z_3)\big)\mapsto 
\big((\tilde{X}_1,\tilde{X}_2,\tilde{X}_3),(\tilde{Y}_1,\tilde{Y}_2,\tilde{Y}_3),(\tilde{Y}_1,\tilde{Y}_2,\tilde{Y}_3)\big),
\end{gathered}\\
\notag
\tilde{X}_1={X_1}+\frac{({\qt_1}-1)}{{\qt_1}}{Y_1},\\
\notag
\tilde{X}_2={X_2}
-\frac{({\qt_1}-1)^2}{(\qt_1)^2} {Y_1}+\frac{({\qt_1}-1)}{{\qt_1}} {Y_3}
-\frac{({\qt_1}-1)}{{\qt_1}}{Z_1}-\frac{({\qt_1}-1)}{{\qt_1}} {Z_3},\\
\notag
\tilde{X}_3={X_3},\qqquad \tilde{Y}_1=\frac{1}{{\qt_1}}{Y_1},\\
\notag
\tilde{Y}_2=({\qt_1}-1)^2 {X_1}+(1-{\qt_1}) {X_2}+(1-{\qt_1}) {X_3}+{\qt_1} {Y_2}+(1-{\qt_1}) {Z_2},\\
\notag
\tilde{Y}_3={Y_3},\qqquad
\tilde{Z}_1=\frac{({\qt_1}-1)}{{\qt_1}} {Y_1}+{Z_1},\qqquad
\tilde{Z}_2=(1-{\qt_1}) {X_1}+{Z_2},\qqquad
\tilde{Z}_3={Z_3}.
\end{gather}

Similarly, in the case~\eqref{ncase2} we obtain
\begin{gather}
\label{lmqt2}
\begin{gathered}
\dft \tm\big|_{(\fp,\fp,\fp)}\cl\mathbb{C}^3\times\mathbb{C}^3\times\mathbb{C}^3
\to\mathbb{C}^3\times\mathbb{C}^3\times\mathbb{C}^3,\qquad
\fp=(1-\qt_2,0,\qt_2),\qquad\qt_2\in\bC,\quad\qt_2\neq 0,\\
\big((X_1,X_2,X_3),(Y_1,Y_2,Y_3),(Z_1,Z_2,Z_3)\big)\mapsto 
\big((\tilde{X}_1,\tilde{X}_2,\tilde{X}_3),(\tilde{Y}_1,\tilde{Y}_2,\tilde{Y}_3),(\tilde{Y}_1,\tilde{Y}_2,\tilde{Y}_3)\big),
\end{gathered}\\
\notag
\tilde{X}_1={X_1}-\frac{({\qt_2}-1)}{{\qt_2}} {Y_3}
+\frac{({\qt_2}-1)}{{\qt_2}} {Z_2}+\frac{({\qt_2}-1)}{{\qt_2}} {Z_3},\\
\notag
\tilde{X}_2={X_2}+(1-{\qt_2}) {Y_2},\qqquad \tilde{X}_3={X_3},\\
\notag
\tilde{Y}_1=
\frac{({\qt_2}-1)}{{\qt_2}} {X_1}+\frac{({\qt_2}-1)}{{\qt_2}} {X_3}
+\frac{1}{{\qt_2}}{Y_1}
-\frac{({\qt_2}-1)^2}{(\qt_2)^2} {Y_3}
+\frac{({\qt_2}-1)}{{\qt_2}} {Z_1}
+\frac{({\qt_2}-1)^2}{(\qt_2)^2} {Z_2}
+\frac{({\qt_2}-1)^2}{(\qt_2)^2} {Z_3},\\
\notag
\tilde{Y}_2={\qt_2} {Y_2},\qqquad\tilde{Y}_3={Y_3},\\
\notag
\tilde{Z}_1=({\qt_2}-1) {X_2}-({\qt_2}-1)^2{Y_2}+{Z_1},\qqquad
\tilde{Z}_2=(1-{\qt_2}) {Y_2}+{Z_2},\qqquad
\tilde{Z}_3={Z_3}.
\end{gather}

Finally, in the case~\eqref{ncase3} one gets a very simple map
\begin{gather*}
\dft \tm\big|_{(\fp,\fp,\fp)}\cl\mathbb{C}^3\times\mathbb{C}^3\times\mathbb{C}^3
\to\mathbb{C}^3\times\mathbb{C}^3\times\mathbb{C}^3,\qquad
\fp=(0,0,\qt_3),\qquad\qt_3\in\bC,\qquad\qt_3\neq 0,\\
\big((X_1,X_2,X_3),(Y_1,Y_2,Y_3),(Z_1,Z_2,Z_3)\big)\mapsto 
\big((\tilde{X}_1,\tilde{X}_2,\tilde{X}_3),(\tilde{Y}_1,\tilde{Y}_2,\tilde{Y}_3),(\tilde{Y}_1,\tilde{Y}_2,\tilde{Y}_3)\big),\\
\tilde{X}_1=X_1,\qqquad
\tilde{X}_2=X_2,\qqquad \tilde{X}_3={X_3},\\
\tilde{Y}_1=\frac{1}{\qt_3}{Y_1},\qqquad
\tilde{Y}_2={\qt_3} {Y_2},\qqquad\tilde{Y}_3={Y_3},\qqquad
\tilde{Z}_1={Z_1},\qqquad
\tilde{Z}_2={Z_2},\qqquad
\tilde{Z}_3={Z_3}.
\end{gather*}
\end{example}

\begin{remark}
\label{rresul}
It was noticed in~\cite{ikkp21} that 
the linear tetrahedron map $\dft{\tm}\big|_{({\fp},{\fp},{\fp})}$ 
described in Corollary~\ref{ctaaa} can be regarded as a linear approximation of
a nonlinear tetrahedron map~${\tm}\cl\md^3\to\md^3$ at a point $({\fp},{\fp},{\fp})\in\md^3$
satisfying ${\tm}\big(({\fp},{\fp},{\fp})\big)=({\fp},{\fp},{\fp})$.
Relations with the local Yang--Baxter equation and Lax representations,
which we consider in this section, were not discussed in~\cite{ikkp21}.
As shown in this section, 
a Lax representation (defined for~${\tm}$ by means of the local Yang--Baxter equation)
helps to find points $({\fp},{\fp},{\fp})\in\md^3$ 
satisfying ${\tm}\big(({\fp},{\fp},{\fp})\big)=({\fp},{\fp},{\fp})$.

The map~\eqref{lmct} was obtained in~\cite{ikkp21}.
The maps~\eqref{lmqt},~\eqref{lmqt1},~\eqref{lmqt2} are new.
\end{remark}

\section{Conclusion}
\label{sconc}

In this paper we have presented several constructions of 
tetrahedron maps, using algebraic and differential-geometric tools.

In particular, we have obtained 
the families of new nonlinear polynomial tetrahedron maps~\eqref{mc1}, \eqref{mc2}
on the space of square matrices, using a matrix refactorisation equation, 
which does not coincide with the standard local Yang--Baxter equation.
In Section~\ref{slint} Liouville integrability has been proved 
for the map~\eqref{mc1} in the case $1\le\bm\le\sz/2$
and for~\eqref{mc2} in the case $3\sz/2\le\bm\le 2\sz-1$.

Also, we have shown how to derive linear tetrahedron maps 
as linear approximations of nonlinear ones, using Lax representations and the differentials 
of nonlinear tetrahedron maps on manifolds.
Applying this construction to nonlinear maps from~\cite{Dimakis,KR},
we have obtained parametric families of new linear tetrahedron maps 
\eqref{lmqt}, \eqref{lmqt1}, \eqref{lmqt2} 
with nonlinear dependence on the parameters $\qt$, $\qt_1$, $\qt_2$. 

Another result is the new nonlinear tetrahedron maps~\eqref{mnK},~\eqref{tmnK}, 
which are matrix generalisations of the map~\eqref{msk} 
from Sergeev's classification~\cite{Sergeev}.
As shown in Proposition~\ref{rnc},
in~\eqref{mnK}--\eqref{tinvmK} one can replace $\mat_{\sz}(\bK)$ by 
an arbitrary associative ring~$\mathcal{A}$ with a unit.
This gives the maps~\eqref{amnK}--\eqref{atinvmK}, 
which can be viewed as tetrahedron maps in noncommuting variables.

Furthermore, new tetrahedron maps~\eqref{tmgr} on an arbitrary group~$\gmd$
have been presented.

Motivated by the results of this paper,
we suggest the following directions for future research:
\begin{itemize}
\item As said above, 
we have proved Liouville integrability of~\eqref{mc1} for $1\le\bm\le\sz/2$
and of~\eqref{mc2} for $3\sz/2\le\bm\le 2\sz-1$.
One can try to prove Liouville integrability of the maps~\eqref{mc1},~\eqref{mc2}
for other values of~$\sz$,~$\bm$.
\item We have obtained the tetrahedron 
maps~\eqref{mc1},~\eqref{mc2}, using the matrix refactorisation equation~\eqref{mrfe} with~\eqref{L12bm}, 
which does not coincide with the matrix local Yang--Baxter equation~\eqref{lybef}.

To derive the maps~\eqref{mc1},~\eqref{mc2}, 
we have used in equation~\eqref{mrfe} with~\eqref{L12bm} the matrix-function~\eqref{Lsz} 
of the form~\eqref{Lsz11}. It would be interesting to obtain other tetrahedron maps 
from equation~\eqref{mrfe} with~\eqref{L12bm}, using matrix-functions~\eqref{Lsz} of other forms.

\item It is well known that Yang--Baxter and tetrahedron maps are closely related to integrable 
lattice equations (see, e.g.,~\cite{Pavlos2019,Pavlos2012,Kassotakis-Tetrahedron,pap-Tongas}).
As discussed above, in this paper we have constructed several families of new nonlinear tetrahedron maps, 
some of which are Liouville integrable, 
and also we have described a procedure to derive new linear tetrahedron maps as linear approximations of nonlinear ones.
It would be interesting to study possible relations of the obtained tetrahedron maps with integrable lattice equations, using the methods described in~\cite{Pavlos2012,Pavlos2019,Kassotakis-Tetrahedron,pap-Tongas} and references therein. 
In particular, it would be interesting to understand whether the above-mentioned linear approximations of nonlinear tetrahedron maps correspond to some sort of linear approximations of nonlinear lattice equations.

\end{itemize}

\section*{Acknowledgements} 

The work on Sections~\ref{sintr},~\ref{stmrmr},~\ref{stmgr}
was supported by the Russian Science Foundation (grant No. 21-71-30011).

The work on Sections~\ref{preliminaries},~\ref{sdtm},~\ref{sconc} 
was carried out within the framework of a development programme for the Regional Scientific and Educational Mathematical Centre of the P.G. Demidov Yaroslavl State University with financial support from the Ministry of Science and Higher Education of the Russian Federation (Agreement on provision of subsidy from the federal budget No. 075-02-2022-886).

We would like to thank A.~Doliwa, V.A.~Kolesov, M.M.~Preobrazhenskaia, S.M.~Sergeev, 
and D.V.~Talalaev for useful discussions.

\end{document}